\documentclass[sigconf, nonacm, balance=false]{acmart}

\usepackage{dsfont}
\usepackage{algorithm}
\usepackage[noend]{algpseudocode}
\usepackage[textsize=small]{todonotes}
\usepackage{tikz}
\usetikzlibrary{trees}
\usepackage{hyperref}
\usetikzlibrary{positioning, arrows.meta}

\newcommand{\err}[0]{\textnormal{err}}

\newtheorem{theorem}{Theorem}
\newtheorem{definition}{Definition}
\newtheorem{proposition}{Proposition}
\newtheorem{observation}{Observation}

\begin{document}

\title[Differentially Private Release of Hierarchical Origin/Destination Data with a TopDown Approach]{Differentially Private Release of Hierarchical Origin/Destination Data with a TopDown Approach}


\author{Fabrizio Boninsegna}
\orcid{1234-5678-9012}
\affiliation{%
  \institution{University of Padova}
  \city{Padova}
  \state{Veneto}
  \country{Italy}}
\email{fabrizio.boninsegna@phd.unipd.it}

\author{Francesco Silvestri}
\orcid{0000-0002-9077-9921}
\affiliation{%
  \institution{University of Padova}
  \city{Padova}
  \state{Veneto}
  \country{Italy}}
\email{francesco.silvestri@unipd.it}

\begin{abstract}
  This paper presents a novel method for generating differentially private tabular datasets for hierarchical data, specifically focusing on origin-destination (O/D) trips. 
  The approach builds upon the TopDown algorithm, a constraint-based mechanism developed by the U.S. Census to incorporate invariant queries into tabular data.
  O/D hierarchical data refers to datasets representing trips between geographical areas organized in a hierarchical structure (e.g., region $\rightarrow$ province $\rightarrow$ city). 
 The proposed method is designed to improve the accuracy of queries covering broader geographical areas, which are derived through aggregation. 
 This feature provides a “zoom-in” effect on the dataset, ensuring that when zoomed back out, the overall picture is preserved. Furthermore, the approach aims to reduce false positive detection.
 These characteristics can strengthen practitioners' and decision-makers' confidence in adopting differential privacy datasets.
 The main technical contribution of this paper includes a novel TopDown algorithm that employs constrained optimization with Chebyshev distance minimization, with theoretical guarantees on the maximum absolute error. 
  Additionally, we propose a new integer optimization algorithm that significantly reduces the incidence of false positives.
  The effectiveness of the proposed approach is validated using real-world and synthetic O/D datasets, demonstrating its ability to generate private data with high utility and a reduced number of false positives. Our experiments focus on O/D datasets with a single trip as a unit of privacy: nevertheless, the proposed approach supports other units of privacy and also applies to any tabular data with a hierarchical structure.
\end{abstract}


\maketitle

\section{Introduction}
\begin{figure}[t]
    \centering
    \includegraphics[scale = 0.25]{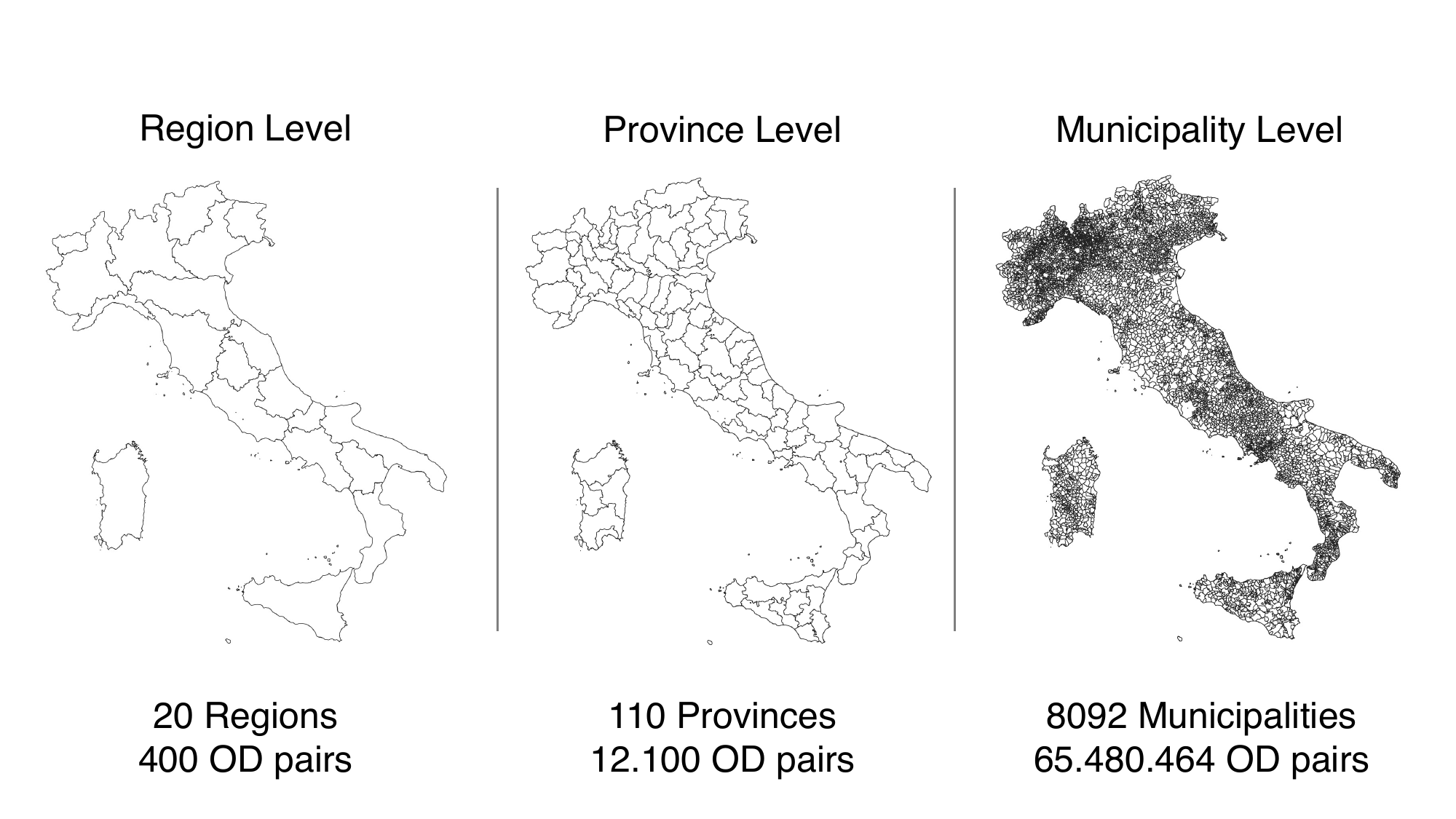}
    \caption{\small The first three hierarchical levels of Italy, according to ISTAT.}
    \label{fig:italy}
\end{figure}
Origin-destination (O/D) data plays a crucial role in policy planning, especially in today's context, where official statistical agencies release detailed trip records. These records include origin and destination areas, along with various trip attributes such as mode of travel and trip purpose. These detailed data are important for a wide range of planning purposes, from transportation planning \cite{lee2022travel} to epidemic modeling \cite{gomez2019impact}, and are essential for understanding and managing the flow of people and goods in various settings. However, the release of mobility data poses significant privacy risks. 
Individuals’ movement patterns can be sensitive information, potentially revealing personal habits and frequented locations. 
Differential Privacy (DP) \cite{dwork2006differential}  provides a rigorous solution to this challenge. It involves introducing randomness to data in a controlled manner, ensuring that the privacy of individuals in a dataset is preserved while still allowing meaningful analysis. 
In this paper, we focus on differentially private release of O/D trips with a geographic hierarchy, focusing on obtaining a "zoom-in" effect (which we call \emph{hierarchical accuracy}) and reducing false positives. We will focus on the case where each user contributes one trip and specifically on \emph{bounded} differential privacy \cite{kifer2011no}, where the total number of users is an invariant and does not change under DP.

\paragraph{The Geographic Hierarchy}
A geographic hierarchy with $g$ levels can be represented as a tree that captures the relationships between these geographic levels. Each node in the tree corresponds to a geographic area, while its children represent subareas into which that area is divided.
For example, the Italian Institute of Statistics (ISTAT)\footnote{https://www.istat.it/it/archivio/222527} organizes Italy into a geographic hierarchy consisting of regions, provinces, municipalities, section areas, and census sections (with the first three levels illustrated in Figure \ref{fig:italy}). The regions represent the largest areas (level 1), and each unit at level $j$ is entirely contained within exactly one unit at level $j-1$.
An O/D dataset of trips $D = \{(u_i, v_i)\}_{i=1,\dots, n}$, where $u_i, v_i$ denote the origin and destination at level $g$ of the user $i$, enables practitioners to compute marginal geographic queries. For example, given an O/D dataset that contains trips between municipalities, it is possible to derive flows (i.e., the number of O/D trips in the dataset) between pairs of provinces or regions. 

\paragraph{Reducing False Positives} In the O/D dataset scenario, the basic mechanism to ensure DP is the addition of Laplace or Gaussian noise to \emph{all} the O/D flows, regardless of the existence of a trip in the data. However, since these datasets are typically sparse (e.g., for trips between cities within a country, only a small subset of possible trips is usually reported), this approach can result in a pathological behavior of false positives, a phenomenon well-documented in the literature~\cite{cormode2011differentially, cormode2012differentially, aumuller2021differentially}.
False positives in data analysis refer to instances where there are O/D trips in the DP data, even though they were not present in the real data.
For sparse datasets, minimizing false positives is essential to preserve their inherent sparsity, ensuring that the released dataset remains compact in size. 
Furthermore, reducing false positives can improve confidence in adopting such solutions:
a dataset cluttered with numerous insignificant trips (with small flows) is likely to be perceived as low quality by decision-makers, even if the flow of false positive trips might not be statistically significant.

\paragraph{The Hierarchical Accuracy} A geographic hierarchy enables practitioners to aggregate differentially private noisy data flows to obtain insights about broader geographic regions.
However, this approach leads to a decrease in precision due to noise accumulation, which significantly reduces the accuracy of the dataset at larger scales.
This outcome contrasts with the intuitive expectation that, for general purpose private data, statistics become less accurate as one “zooms-in” to smaller regions.
One potential solution is to generate multiple DP datasets, each tailored to a specific geographic level.
 However, this approach has two major drawbacks: it requires practitioners to manage multiple datasets, and the results derived from these datasets may lack consistency. 
 Therefore, there is a clear need for a single tabular dataset 
 that provides hierarchical accuracy, offering higher precision for broader geographic regions while maintaining consistency and usability.

\vspace{\baselineskip}
Our goal is to derive a mechanism to release a unique O/D tabular dataset $\tilde{D}$ at level $g$ with
\begin{enumerate}
\item (\emph{Privacy}) The dataset $\tilde{D}$ is $(\epsilon, \delta)$-differentially private.
\item (\emph{Hierarchical accuracy}) 
The dataset $\tilde{D}$ has a "zoom-in" effect on accuracy. Hence, it has higher accuracy for queries that cover larger geographic areas compared to those covering smaller areas \footnote{In the sense that the estimates related to area A are more accurate than those related to any area nested within A.}. 
\item (\emph{Reduce False Positives}) The differentially private algorithm is designed to reduce the occurrence of false positives.
\end{enumerate}
Our algorithm \texttt{InfTDA} rigorously satisfies the first two points and has been experimentally shown to reduce false positives in real and synthetic data sets. \texttt{InfTDA} is built upon the TopDown Algorithm (\texttt{TDA}) developed by the US Census \cite{abowd20222020} and used to publish the 2020 US census. The main idea of \texttt{TDA} is to use an iterative optimization approach along the hierarchy to ensure hierarchical accuracy in practice. 
\texttt{InfTDA} distinguishes itself from \texttt{TDA} through a modified optimization approach and a specific focus on addressing false positives, which are effectively minimized using a heuristic integrated into the optimization process.

\subsection{Our Results} 
Our algorithm \texttt{InfTDA} builds upon \texttt{TDA} with a key modification: it uses the Chebyshev distance (which is the $\ell_\infty$ distance, so the name of our algorithm) as the objective function in the constrained optimization problem, instead of the Euclidean distance in \texttt{TDA}. 
Intuitively, minimizing Chebyshev distance is simpler than Euclidean, as the former can be minimized with a linear program.
More importantly, to the best of our knowledge, with this choice we are the first to provide an upper bound for the utility of a \texttt{TDA}-like mechanism \footnote{We intend any mechanism that uses an iterative optimization approach along a hierarchy, regardless of the optimization or the differentially private mechanism used.}. 
Consequently, we propose \texttt{IntOpt} as an optimizer for the Chebyshev distance that works entirely in the integer domain and is designed with a heuristic to reduce false positives.
Finally, we show that any \texttt{TDA}-like algorithm works for a specific tree data structure having some well-defined properties, which we call \emph{non-negative hierarchical tree}. With this generalization, we ensure the broad applicability of \texttt{InfTDA} to any tabular data that can be represented as a non-negative hierarchical tree, making it highly generalizable \footnote{For example, in a healthcare dataset encompassing diseases and user characteristics, a possible hierarchy might start with counts of diseases, segmented further by gender, then by age groups, and so on. }.

The main results of this paper are
\begin{enumerate}
    \item A demonstration that any O/D dataset can be parsed into a non-negative hierarchical tree, and so it is suitable for a TopDown algorithm.
    \item A theoretical analysis of the accuracy of \texttt{InfTDA}. The chosen accuracy is the \emph{maximum absolute error} at each level of the hierarchy. 
    \item We propose \texttt{IntOpt}, a fast integer-constrained optimization algorithm that minimizes Chebyshev distance and includes a heuristic to reduce false positives in practical scenarios.
\end{enumerate}
We evaluated \texttt{InfTDA} using both real-world data (a dataset of O/D commuting flows in Italy) and synthetic O/D data. The results show that its utility is not worse than \texttt{TDA} while being faster, simpler, and generates a dataset with fewer false positives.
In contrast to \texttt{TDA}, our \texttt{InfTDA} has a theoretical upper bound for the utility, which we now informally show for O/D datasets.

\begin{theorem}[Informal version of utility of \texttt{InfTDA}] Given a O/D dataset with $g$ geographic levels. \texttt{InfOPT} with constant probability returns a differentially private tabular dataset, with maximum absolute error at most $\tilde{O}(\sqrt{\ell^3 g})$, for O/D flows with origin and destination at level $\ell \in \{0,\dots, g\}$.
\end{theorem}

A more formal version for any non-negative hierarchical tree is stated in Theorem \ref{th: InfTDA}. 

\subsection{Previous Work}
\paragraph{Histogram in DP} Releasing a differentially private O/D dataset is essentially equivalent to the release of differentially private histograms, extensively explored in existing literature \cite{xu2013differentially, zhang2014towards, suresh2019differentially}. In fact, $D=\{(u_i, v_i)\}_{i=1,\dots,n}$ has a classical histogram representation where each bin counts the occurrence of an O/D pair in $D$. This can be seen as a tabular representation with origin, destination, and flow (i.e., counts) as columns.
The main strategy involves adding Laplace noise to achieve pure differential privacy, as detailed by Dwork et al. \cite{dwork2006calibrating}, or incorporating Gaussian noise for approximate differential privacy \cite{balle2018improving, canonne2020discrete}. The latter has found practical application in the US Census \cite{abowd20222020} as Gaussian tail bounds provide a better trade-off between privacy and outliers. These methods return an unbiased estimator of the histogram, but are space-inefficient when the histogram is sparse, as it is necessary to generate independent noise for every count. To address this challenge, more sophisticated algorithms have been developed specifically for sparse histograms. The stability-based method \cite{korolova2009releasing, vadhan2017complexity,swanberg2023dp} relies on thresholding private counts, getting a biased estimator without false positives. For the sensitivity one scenario, Desfontaines et al. \cite{desfontaines2022differentially} have developed a mechanism that is optimal to reduce the number of false negatives. A method to obtain an unbiased estimator of the histogram has been developed by Cormode et al. 
\cite{cormode2012differentially} using priority sampling at the price of increasing the expected error. 
Aumüller et al. \cite{aumuller2021differentially} developed a biased estimator for the counts that matches a lower bound up to a constant factor.
All the aforementioned approaches return a differentially private histogram representation of $D$, mostly by matching lower bounds for counts \cite{vadhan2017complexity}. However, when range queries are important (i.e., sum of counts), these mechanisms do not offer strong upper bounds due to noise propagation.

\paragraph{Constrained Optimization} The pioneering study by Hay et al. \cite{Boosting_accuracy} was the first to observe that the accuracy of range queries in differential privacy could be improved through a post-processing stage involving constrained optimization. This was achieved using the Hierarchical mechanism, which is based on the construction of a tree of range queries. In this approach, the outputs of differentially private queries are post-processed such that the result of each query node is derived from an aggregation of the results from its child nodes. The authors also provided an analytical solution to the optimization problem, however it only applies to queries in the real domain, which makes it not appropriate to release counts. 

Fioretto et al. \cite{fioretto2018constrained} further advanced the field, particularly for mobility data, by developing the Constrained Based Differential Private (\texttt{CBDP}) mechanism. This mechanism specifically addresses the challenge of releasing mobility data with differential privacy, particularly for On-Demand multi-modal transit systems. Building upon the Hierarchical mechanism, \texttt{CBDP} enhances it by integrating the \emph{non-negativity} constraint (i.e., only positive flows are accepted; negative flows may occur due to injection of randomness) and generalizes it by allowing more general constraints than only hierarchical ones. However, this involves solving a unique constrained optimization problem with a potentially prohibitively large number of variables and constraints, especially for O/D data. 

The most significant utilization of constrained optimization has been observed in the Disclosure Avoidance System TopDown algorithm (which we called it \texttt{TDA}) implemented by the US Census Bureau to publish the 2020 USA census data \cite{abowd20222020}. The goal of \texttt{TDA} was to release population histograms, including ethnicity and age distributions, for each census section while maintaining a fixed number of people per state. To achieve this, the US Census Bureau devised a TopDown algorithm that follows a geographic hierarchy that they called \emph{geographic spine} (Nations $\rightarrow$ Regions $\rightarrow$ Divisions $\rightarrow$ States $\rightarrow$ Counties etc.) \cite{bureau2020disclosure}. 
Starting with the creation of a differentially private tabulated dataset at the national level, the TopDown approach then proceeds to process data for regions. 
To ensure coherence, a constraint optimization problem is solved during post-processing, aiming for the regional tabulated data to be consistent with that of the nation. 
Essentially, the aggregated attributes of the regions must align with the national attributes. 
The process is repeated iteratively up to the final level of the geographic spine, ending in the release of micro-data at the census section level.
The advantage of a TopDown approach is twofold: it splits the optimization problem into many more feasible problems, and it mimics the sparsity of the data. The latter advantage is an inherent characteristic of the TopDown approach, as identifying false positives is more straightforward in large aggregated datasets. 
Once an attribute is determined to not exist in a geographic area, it is inferred that the same attribute will also be absent in any subdivision of that area. This inference avoids the introduction of noise in these areas throughout the TopDown process.

\paragraph{Objective Function} All the previously discussed algorithms aim to solve a constrained optimization problem by returning a vector that is as close as possible to the differentially private estimate, measured in terms of the Euclidean distance, referred to as the $\ell_2$ distance. Hay et al. \cite{Boosting_accuracy} raised the possibility of minimizing the $\ell_1$ distance (the sum of absolute errors between the differentially private estimate and the released vector). Since this approach does not guarantee a unique solution, the authors chose the $\ell_2$ minimization, as it ensures uniqueness. \texttt{TDA} uses a \emph{weighted non-negative} least squares optimization, meaning that a weight could be assigned to each absolute error. In particular, the US Census opted to use the inverse variance of the differentially private random variables as weights. A similar weighted approach was used by Fioretto et al. \cite{fioretto2018constrained}. 
\newline
\newline
We chose to minimize the Chebyshev distance for two purposes: it produces a simple integer constraint minimization problem, and it gives theoretical guarantees. In contrast, \texttt{TDA} uses a complex two-step optimization algorithm, first solving the problem in the real space with convex optimization and then performing the best integer rounding with linear programming. 

\subsection{Structure of the Paper}
In Section~\ref{sec: prelim}, we introduce the core principles of differential privacy, highlighting two key mechanisms: the Gaussian mechanism and the Stability Histogram. Additionally, we define the non-negative hierarchical tree. In Section~\ref{sec:tree}, we define the notation used for the hierarchical structure of O/D data, its reformulation as a non-negative hierarchical tree, and the utility metric used in our analysis. The tree reformulation is then used to formulate our algorithm \texttt{InfTDA} in Section~\ref{sec: top down mechanism}, making it of broad applicability. This is followed by a theoretical analysis of \texttt{InfTDA} and the introduction of \texttt{IntOpt}, an optimizer tailored to reduce false positives.
The discussion concludes in Section \ref{section: experiments}, which presents an experimental evaluation of \texttt{InfTDA} compared to baselines and other \texttt{TDA}-like mechanisms on real-world and synthetic datasets, highlighting the advantages of the proposed methodology.
\section{Preliminaries}
\label{sec: prelim}
\subsection{Differential Privacy}
This section introduces the key concepts of differential privacy relevant to this article. In our discussion, we focus on the notion of privacy known as \emph{bounded} \cite{kifer2011no}. In this framework, two datasets, $D$ and $D'$, are neighbors (that is, $D\sim D'$) if one can be obtained from the other by substituting a single user. Thus, the name bounded refers to the fact that neighboring datasets have a fixed size. In contrast, within the framework of \emph{unbounded} differential privacy \cite{kifer2011no}, two neighboring datasets may differ by either the addition or removal of a single user, allowing their sizes to vary.
In the following, we present the formal definition of differential privacy.
\begin{definition}[Differential Privacy (DP) \cite{dwork2014algorithmic}] Given $\varepsilon>0$ and $\delta \in [0,1)$. A randomized mechanism $\mathcal{M}: \mathcal{D}\rightarrow \mathcal{R}$ satisfies $(\varepsilon, \delta)$-differential privacy if for any two neighboring datasets $D$,$D'\in \mathcal{D}$ and for any subset of outputs $S\subseteq \mathcal{R}$ it holds that
\begin{equation*}
    \text{Pr}[\mathcal{M}(D)\in S]\leq e^{\varepsilon}\text{Pr}[\mathcal{M}(D')\in S] + \delta.
\end{equation*}
\end{definition}
Another definition, which is more suitable to study the injection of Gaussian noise (used in practical applications, such as that used in the US Census release \cite{abowd20222020}) is \emph{zero-Concentrated Differential Privacy}. This definition provides a tighter analysis of the aforementioned mechanism and yields simpler composition results.
\begin{definition}[zero-Concentrated Differential Privacy (zCDP)\cite{bun2016concentrated}] 
    Given $\rho >0 $. A randomized mechanism $\mathcal{M}: \mathcal{D}\rightarrow \mathcal{R}$ satisfies $\rho$-zCDP if for any two neighboring datasets $D$,$D'\in \mathcal{D}$ and any $\alpha>1$
    \begin{equation*}
        D_{\alpha}(\mathcal{M}(D)||\mathcal{M}(D'))\leq \alpha \rho,
    \end{equation*}
    where $ D_{\alpha}(\mathcal{M}(D)||\mathcal{M}(D'))$ is the $\alpha$-Rényi divergence.
\end{definition}
Any $\rho$-zCDP mechanism satisfies also $(\varepsilon, \delta)$-DP.

\begin{lemma}[From $\rho$-zCDP to $(\varepsilon, \delta)$-DP (Lemma 21 in \cite{bun2016concentrated})]
\label{lemma: from DP to zCDP}
Let $\mathcal{M}: \mathcal{D}\rightarrow \mathcal{R}$ satisfy $\rho$-zCDP. Then $\mathcal{M}$ satisfies $(\varepsilon, \delta)$-DP for all $\delta\in (0,1)$ and
\begin{equation*}
    \varepsilon = \rho + 2\sqrt{\rho \log(1/\delta)}.
\end{equation*}
\end{lemma}
This lemma is useful for converting $\rho$-zCDP guarantees into $(\varepsilon,\delta)$-DP, which are more commonly used in real-world applications.
A significant benefit of employing differential privacy is its resilience to privacy degradation, regardless of the application of any post-processing functions.
\begin{lemma}[Post-Process Immunity (Lemma 8 \cite{bun2016concentrated})] Let $\mathcal{M}:\mathcal{D}\rightarrow\mathcal{R}$ and $f: \mathcal{R}\rightarrow\mathcal{R}'$ be an arbitrary (also randomized) mapping. Suppose $\mathcal{M}$ satisfies $\rho$-zCDP. Then $f \circ \mathcal{M}: \mathcal{D}\rightarrow \mathcal{R}'$ satisfies $\rho$-zCDP.
\end{lemma}
 This characteristic is crucial when data are released to practitioners. If a dataset $\tilde{D}$ is generated using a differentially private algorithm, then running any query on this dataset will not compromise its privacy. Another important property of differential privacy is the composition of privacy budgets, which enables the computation of differential privacy guarantees for the composition of several private algorithms. 

\begin{lemma}[Composition (from Lemma 7 in \cite{bun2016concentrated})] Let $\mathcal{M}: \mathcal{D}\rightarrow\mathcal{R}$ and $\mathcal{M}': \mathcal{D}\rightarrow\mathcal{R}'$ be randomized algorithms. Suppose $\mathcal{M}$ satisfies $\rho$-zCDP and $\mathcal{M}'$ satisfies $\rho'$-zCDP. Define $\mathcal{M}'': \mathcal{D}\rightarrow \mathcal{R}\times \mathcal{R}'$ by $\mathcal{M}''(x) = (\mathcal{M}(x), \mathcal{M}'(x))$. Then $\mathcal{M}''$ satisfies $(\rho +\rho')$-zCDP.
\end{lemma}
Since O/D data can be visualized as a histogram covering all possible O/D trips, we shift our focus to the release of histograms under differential privacy. The histogram for a dataset $\mathcal{X}^{n}$, where $\mathcal{X}$ denotes the data universe (e.g., all possible O/D pairs) and $n$ is the number of users, is generated by a counting query $q: \mathcal{X}^{n} \rightarrow \mathbb{N}^{|\mathcal{X}|}$. This query outputs the absolute frequency of each row in the dataset. 
A common paradigm for approximating these functions with differentially private mechanisms is through \emph{additive noise mechanisms} calibrated to the $p$-\emph{global sensitivity} $\text{GS}_{p}(q)$ of the function, which is defined as the maximum absolute $\ell_p$ distance 
$$\text{GS}_p (q) = \sup_{D\sim D'}||q(D)-q(D')||_p,$$
where the $\sup$ is taken over two neighboring datasets. If a user contributes at most to $m$ different trips in an O/D dataset then, under bounded privacy, we have $\text{GS}_1(q) = 2m$ and $\text{GS}_2(q) = \sqrt{2m}$.\footnote{For unbounded privacy, we instead have $\text{GS}_1(q) = m$ and $\text{GS}_2(q) = \sqrt{m}$. The factor of $2$ in bounded privacy arises from the definition of neighboring datasets, where $D’$ is considered a neighbor of $D$ if it differs from $D$ by both the addition and removal of two distinct users (i.e., a substitution).} The first calibrates the additive noise from a Laplace distribution, while the second calibrates the noise from a Gaussian distribution. We now illustrate two mechanisms for the release of $(\varepsilon, \delta)$-DP histograms.

\subsubsection{Discrete Gaussian Mechanism} 
If the data universe $\mathcal{X}$ is finite and known, we can achieve zCDP (so approximate DP due to Lemma \ref{lemma: from DP to zCDP}) by adding Gaussian noise to each count.
\begin{theorem}[Discrete Gaussian Mechanism \cite{canonne2020discrete}] 
\label{th: discrete gaussian}
Let $q: \mathcal{X}^n \rightarrow \mathbb{N}^{|\mathcal{X}|}$ be a counting query. The discrete Gaussian mechanism applied to a counting query $q(D)$ consisting of the injection of discrete Gaussian noise
\begin{equation*}
    \tilde{q}(D) = q(D) + Z \qquad Z \sim \mathcal{N}_{\mathbb{Z}}\bigg(0, \frac{{\textnormal{GS}}_2 (q)^2}{2\rho}\bigg )^{|\mathcal{X}|}
\end{equation*}
is $\rho$-zCDP.
\end{theorem}
The precision of the mechanism is not worse than that of its continuous counterpart \cite{bun2016concentrated}
\begin{corollary}[Corollary 9 \cite{canonne2020discrete}]
\label{corollary: discrete Gaussian tail}
Let $Z\sim \mathcal{N}_{\mathbb{Z}}(0, \sigma^2)$. Then $\textnormal{Var}[Z]\leq \sigma^2$ and $\textnormal{Pr}[Z\geq t]\leq e^{-\frac{t^2}{2\sigma^2}}$ for any $t\geq 0$.  
\end{corollary}
\subsubsection{Stability-Based Histogram} 
If the data universe $\mathcal{X}$ is unknown, infinite, or very large, we can apply Laplace noise to each positive count of the histogram as long as noisy counts smaller than a certain threshold are set to zero. 
\begin{theorem}[\texttt{SH}-Stability-Based Histogram \cite{bun2019simultaneous}] Let $q: \mathcal{X}^n \rightarrow \mathbb{N}^{|\mathcal{X}|}$ be a counting query of $1$-global sensitivity equal to 2. The algorithm that first applies Laplace noise to positive queries
\begin{equation*}
    \tilde{q}(x_i) = q(x_i) + Z \qquad \forall x_i \in \mathcal{X} \,: \,q(x_i)>0, \, \text{and}\, Z\sim e^{-\frac{2}{\varepsilon}|x|}
\end{equation*}
and then maps to zero the noisy counts under $t = 1 + \frac{2\log(2/\delta)}{\varepsilon}$, is $(\varepsilon, \delta)$-differentially private.
\end{theorem}
It is important to stress that this mechanism \emph{does not return false positives by construction} since it only injects noise into positive counts. However, it returns a biased estimator due to thresholding, which might be useless for estimating aggregate queries. For example, if all counts fall below the threshold, any aggregate query would result in zero.

\subsection{Non-Negative Hierarchical Tree}
 The \texttt{InfTDA} algorithm and the analysis provided in this paper are formulated for a tree data structure satisfying a few properties, ensuring broad applicability beyond O/D data. A tree $\mathcal{T}$ of depth $T$ is a tuple $\mathcal{T}=(V, E)$:  $V=\cup_{\ell=0}^{T}V_\ell$ is the set of nodes, with $V_\ell$ denoting the set of nodes at level $\ell \in [T]$, for $[T]=[0,\dots, T]$; $E$ is the set of edges between consecutive levels. A node at level $\ell$ is indicated with $u_\ell \in V_{\ell}$. The set of children to a node is indicated by a function $\mathcal{C}: V_{\ell}\rightarrow 2^{V_{\ell+1}}$ for any level $\ell\in [T]$. Each node of the tree has an attribute $q(u_\ell)$. 
 The specific tree we are interested in this paper is defined as follows.
 \begin{definition}[Non-Negative Hierarchical Tree] 
 \label{def: non-negative hierarchical tree}
 A tree $\mathcal{T} = (V, E)$ is said to be non-negative if it contains non-negative attributes $q(u)\geq 0$ for any $u \in V$. The tree is hierarchical if the attribute of $u$ can be computed as the sum of the attributes of its children $\mathcal{C}(u)$. Specifically, for every  $u \in V$  that is not a leaf, we have:
 \begin{equation}
\label{eq: hierarchical relation}
    q(u) = \sum_{v\in \mathcal{C}(u)}q(v).
\end{equation}
 \end{definition}
In our algorithm, we will use the function ${\bf q}_{\mathcal{C}}(u_\ell)$, which is the vector containing the attributes of the children of $u_\ell$.
For the theoretical analysis of the algorithm, we consider a tree with fixed branching factor $b\in \mathbb{N}$, so that $|V_{\ell}| = b^{\ell}$. 
Given any randomized mechanism $\mathcal M$ applied to the attributes of the tree, the utility metric is defined for each level $\ell \in [T]$ as the \emph{maximum absolute error}
\begin{equation}
\label{eq: max error for non-negative tree}
    \max_{u_\ell \in V_\ell}|\err(u_\ell)| = \max_{u_\ell \in V_\ell}|\mathcal{M}(q(u_\ell))-q(u_\ell)|.
\end{equation}

\section{Tree Structure of O/D Data}
\label{sec:tree}
\begin{figure*}[t]
\begin{tikzpicture}[level/.style={sibling distance=40mm/#1}, node distance=1.5cm]
  \draw[draw=black] (-10,-0.5) rectangle ++(2.5,1);
  \draw[draw=black] (-7,-0.5) rectangle ++(2.5,1);
  \node at (-9.75, -0.25) {$\scriptstyle u_\ell$};
  \draw[-Latex, bend left=45] (-8.75, 0) to node[midway, above] {} (-5.75, 0);
  \node at (-10.5, -0.5) {(a)};
  \node at (-6.75, -0.25) {$\scriptstyle v_\ell$};
  \node at (-8.75, 1) {\textbf{Origin}};
  \node at (-5.75, 1) {\textbf{Destination}};

  \draw[draw=black] (-10,-2) rectangle ++(2.5,1);
  \draw[draw=black] (-7,-2) rectangle ++(2.5,1);
  \draw[black, -] (-5.75, -2) -- (-5.75, -1);
  \draw[-Latex, bend left=45] (-8.75, -1.5) to node[midway, above] {} (-6.25, -1.5);
  \draw[-Latex, bend right=45] (-8.75, -1.5) to node[midway, above] {} (-5, -1.5);
  \node at (-9.75, -1.75) {$\scriptstyle u_\ell$};
  \node at (-10.5, -1.75) {(b)};
  \node at (-6.65, -1.75) {$\scriptstyle v_{\ell+1,0}$};
  \node at (-4.85, -1.25) {$\scriptstyle v_{\ell+1,1}$};

  \draw[draw=black] (-10,-3.5) rectangle ++(2.5,1);
  \draw[draw=black] (-7,-3.5) rectangle ++(2.5,1);
  \draw[black, -] (-5.75, -3.5) -- (-5.75, -2.5);
  \draw[black, -] (-8.75, -3.5) -- (-8.75, -2.5);
  \draw[-Latex, bend left=45] (-8.25, -3) to node[midway, above] {} (-6.25, -3);
  \draw[-Latex, bend right=45] (-8.25, -3) to node[midway, above] {} (-5, -3);
  \draw[-Latex, bend left=45] (-9.5, -3) to node[midway, above] {} (-6.25, -3);
  \draw[-Latex, bend right=45] (-9.5, -3) to node[midway, above] {} (-5, -3);
  \node at (-9.65, -3.25) {$\scriptstyle u_{\ell+1,0}$};
  \node at (-8.35, -2.75) {$\scriptstyle u_{\ell+1,1}$};
  \node at (-6.65, -3.25) {$\scriptstyle v_{\ell+1,0}$};
  \node at (-4.85, -2.75) {$\scriptstyle v_{\ell+1,1}$};
  \node at (-10.5, -3.25) {(c)};
  \node at (-2, -0.5) {(d)};
  \node [rectangle, draw] (root){$\scriptstyle q(u_\ell, v_\ell)$}
    child {node [rectangle,draw] (n1) {$\scriptstyle q(u_\ell, v_{\ell+1,0})$}
      child {node [rectangle,draw] (n11) {$\scriptstyle q(u_{\ell+1,0}, v_{\ell+1,0})$}}
      child {node [rectangle,draw] (n12) {$\scriptstyle q(u_{\ell+1,1}, v_{\ell+1,0})$}}
    }
    child {node [rectangle,draw] (n2) {$\scriptstyle q(u_\ell, v_{\ell+1,1})$}
      child {node [rectangle,draw] (n21) {$\scriptstyle q(u_{\ell+1,0}, v_{\ell+1,1})$}}
      child {node [rectangle,draw] (n22) {$\scriptstyle q(u_{\ell+1,1}, v_{\ell+1,1})$}}
    };
  \node at (0,1) {\textbf{Destination Tree}};

\end{tikzpicture}
\caption{\small Example of the two-step construction for the destination tree, represented in the left figure from the top to the bottom. In Figure (a), we have two areas at level $\ell$, $u_\ell$ and $v_{\ell}$, and an arrow with attribute $q(u_\ell, v_\ell)$ indicating the flow between them. In Figure (b), the first step is depicted, the destination area $v_\ell$ is divided into its child areas $v_{\ell+1, 0}$ and $v_{\ell+1,1}$ (in this example, we used a bi-partition). The arrows indicate the cross-level range query of order one. In Figure (c), the last step is depicted, the origin area is divided as well, and the arrows indicate the intra-level query of the finer geographic level $\ell+1$. Figure (d) depicts the destination tree. The links assure hierarchical consistency such that the value of a node can be obtained as the sum of the values of its children.}

\label{fig: tree structure}
\end{figure*}
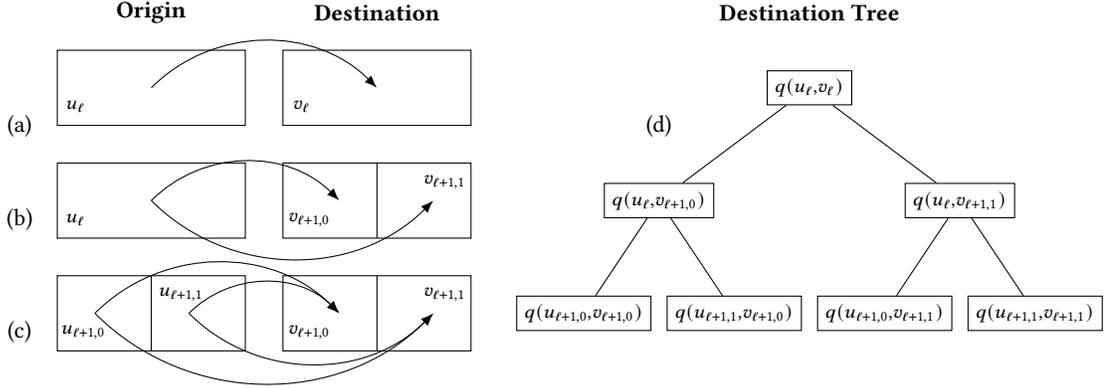
In this section, we show how any O/D dataset can be parsed into two different non-negative hierarchical trees, which we call the origin and destination trees. This reformulation is useful for describing some queries in the dataset.
We start by defining the hierarchy in the geographic space.

\paragraph{Space Partitioning}
Let $X$ be a geographic area (e.g., $X$ represents Italy), and assume that $X$ is hierarchically partitioned into $g$ levels $(P_1, \dots, P_g)$. The dependency among levels is described by the relation $h_{\ell} : P_{\ell} \rightarrow P_{\ell-1}$, which are injections mapping areas at level $\ell$ to the larger areas at level $\ell-1$, for any $\ell \in [1,\dots, g]$. Note that, according to the previous definition, an area at level $\ell'$ is included in only one area at level $\ell<\ell'$.
In our example in Figure \ref{fig:italy}, $X$ is all of Italy, while $P_1$ is the set of regions, $P_2$ is the set of provinces and $P_3$ is the set of municipalities.
With a slight abuse of notation, we write $v_{\ell} \in v_{\ell-1}$ to indicate that area $v_\ell$ is embodied in area $v_{\ell-1}$, so $v_{\ell}\in h_{\ell}^{-1}(v_{\ell-1})$. This holds for any geographical inclusion, so $v_{\ell'}\in v_{\ell}$ if $v_{\ell'}$ is completely contained in the larger region $v_{\ell}$ for $\ell < \ell'$.

\paragraph{The O/D Dataset and Range Queries}
 The dataset we study is a collection of O/D pairs (also called trips) in the finest partition $P_g$. It is represented as $D = \{(u_{g, i}\,,\, v_{g, i})\}_{i = 1,\dots,n}$, for $u_g, v_g \in P_g$. For example, considering the hierarchy in Figure \ref{fig:italy} that stops at the municipality level (e.g., $g=3$), the dataset would contain trips between municipalities.
These O/D pairs can be counted and aggregated to answer marginal queries, which we refer to as \emph{hierarchical range queries}.
\begin{definition}[Hierarchical Range Query]
\label{def: Hierarchical Range Query}
Given two levels $\ell_1, \ell_2 \in [g]$ and two areas $u_{\ell_1} \in P_{\ell_1}$ and $v_{\ell_2} \in P_{\ell_2}$, the hierarchical range query is
	\begin{equation*}
		q(u_{\ell_1}, v_{\ell_2}) = \sum_{u_g \in u_{\ell_1}}\sum_{v_g \in v_{\ell_2}} \sum_{x \in D}\mathds{1}\{x = (u_g\,,\, v_g)\}.
	\end{equation*}
\end{definition}
Here, $\mathds{1}\{x = (u_g\,,\, v_g)\}$ is an indicator function that is used to count how many O/D pairs (that is, the flow) are in the dataset. Meanwhile, $P_0$ represents the entire geographic space $X$.
These queries are conceptually equivalent to SQL \texttt{GROUP BY} followed by \texttt{SUM}, so we will refer to them simply as range queries.
In particular, we are interested in \emph{intra-level} range queries when $\ell_1=\ell_2$, and \emph{cross-level} range queries of order one, when $|\ell_1-\ell_2| = 1$, as they allow us to construct the origin or destination tree, thanks to a \emph{hierarchical consistency}.
\begin{observation}[Hierarchical Consistency]
\label{obs: hierarchical relation}
	Given two levels $\ell_1, \ell_2 \in [g]$ and two areas $u_{\ell_1} \in P_{\ell_1}$ and $v_{\ell_2} \in P_{\ell_2}$, then for any $\ell_1\leq\ell'_1\leq g$ and $\ell_2\leq \ell'_2\leq g$ we have
	\begin{equation*}
		q(u_{\ell_1}, v_{\ell_2}) = \sum_{u_{\ell'_1}\in u_{\ell_1}}\sum_{v_{\ell'_2}\in v_{\ell_2}}q(u_{\ell'_1}, v_{\ell'_2}).
	\end{equation*}
\end{observation}
Following the Italy example, the number of trips from the region Veneto to the region Lombardia has to be the sum of the number of trips among their cities. Another example is that the number of Italians going to Milan must be the sum of the number of trips starting in any Italian region and ending in Milan.
Notice that origins and destinations may belong to two completely different geographic spaces $X,X'$. Before, we considered the case where $X=X'$ and have the same partitions; however, the same arguments can be applied even if $X\neq X'$. For example,  $X$  and  $X'$  could represent two different countries, such as Italy and Germany. Italy’s partitioning may include regions $\to$ provinces $\to$ municipalities, while Germany’s may consist of states $\to$ districts $\to$ municipalities. The dataset might then represent trips between municipalities in Italy and Germany, which can be aggregated at coarser levels within their respective hierarchies. In the next paragraph, we introduce the tree construction in the case $X, X'$ having the same number of partitions $g=g'$.

\paragraph{The Destination and Origin Trees}
We explain the construction for the \emph{destination tree}, the origin tree will follow naively.
The destination tree is a rooted tree that contains information about \emph{intra} and \emph{cross} level range queries of order one. Any node in the tree contains an origin area $u_{\ell_1}$, a destination area $v_{\ell_2}$, and a range query $q(u_{\ell_1}, v_{\ell_2})$ with the property that it can be obtained by summing the queries of its child nodes. 
The root node contains the intra-level query at the zero level, hence the triple $(u_0, v_0, q(u_0, v_0)=n)$. The construction then follows an iterative two-step procedure. Given a node $(u_\ell, v_\ell, q(u_\ell, v_\ell))$:
\begin{enumerate}
	\item create a child node for each finer destination $v_{\ell+1}\in v_{\ell}$ with attribute $(u_\ell, v_{\ell +1}, q(u_{\ell}, v_{\ell+1}))$;
	\item for each child node having destination $v_{\ell+1}$, expand the branch by creating a child node for each finer origin $u_{\ell+1}\in u_{\ell}$ with attribute $(u_{\ell+1}, v_{\ell+1}, q(u_{\ell+1}, v_{\ell+1}))$.
\end{enumerate}
 Each iteration adds two levels in the tree, first by adding a cross-level hierarchical query of order one, then by adding intra-level queries, ending with a tree of $T = 2g+1$ levels. Figure \ref{fig: tree structure} offers an example of the two-step construction of the destination tree. Using the geographic hierarchy of Italy as an example, the destination tree at level 1 represents a histogram of trips from the country level to the regions, level 2 captures trips between regions, level 3 corresponds to trips from regions to provinces, continuing up to level 6, which contains trips between municipalities.
 The origin tree can be obtained similarly by selecting finer origins in step (a). The choice of using the destination or origin tree depends on what the practitioners wish to focus on. If cross-level queries starting from origins belonging to larger areas (e.g., regions) and ending to destinations belonging to smaller areas (e.g., provinces) are more important, then the destination tree is the best choice. In the opposite case, we advise to choose the origin tree. Finally, it is important to state that from the non-negative hierarchical tree, we can obtain the original O/D dataset. This is because the leaves of the tree represent the histogram of the dataset $D$, indicating the absolute frequency with which each O/D pair $(u_g, v_g)$ is observed. 

\begin{lemma}[Relation with Non-Negative Hierarchical Tree] The destination (origin) tree is a non-negative hierarchical tree.
\end{lemma}
\begin{proof}
    Any node of the destination (origin) tree is a tuple of O/D pairs with non-negative attributes defined in Definition \ref{def: Hierarchical Range Query}. Given a father node $(u_\ell, v_{\ell})$, its set of children is $\mathcal{C}(u_\ell, v_\ell) = \{(u_{\ell},v_{\ell+1}): \forall v_{\ell + 1} \in v_{\ell}\}$ (for the origin tree the set of children is $\mathcal{C}(u_\ell, v_\ell) = \{(u_{\ell+1},v_{\ell}) : \forall u_{\ell +1} \in u_\ell\}$). The hierarchical consistency in Observation \ref{obs: hierarchical relation} states that
    \begin{equation*}
        q(u_{\ell}, v_{\ell}) = \sum_{v_{\ell+1}\in v_\ell} q(u_\ell, v_{\ell+1})=\sum_{x \in \mathcal{C}(u_{\ell}, v_{\ell})}q(x),
    \end{equation*}
    which is the hierarchical property in Equation \ref{eq: hierarchical relation}. The analysis applies similarly to the next level of the destination tree (as well as to the origin tree).
\end{proof}

\paragraph{Errors} Let $\mathcal{M}$ be a $(\varepsilon, \delta)$-DP mechanism that acts on range queries. 
We are interested in the \emph{maximum absolute error} for any cross-level range queries of order one and intra-level range queries. Therefore, for $(\ell_1, \ell_2) = \{(0,0), (0,1), (1,1), (1,2), \dots , (g,g)\}$, the error is 
\begin{equation}
\label{eq: error for OD}
	\max_{(u_{\ell_1}, v_{\ell_2}) \in P_{\ell_1} \times P_{\ell_2}} \big|\mathcal{M}(q(u_{\ell_1}, v_{\ell_2})) -q(u_{\ell_1}, v_{\ell_2}) \big|
\end{equation}
Hence, we are interested in the maximum absolute error for any level of the destination tree.
The maximum error defined in Equation \ref{eq: error for OD} can be reformulated as in Equation \ref{eq: max error for non-negative tree}.

\section{The Top Down Algorithm}
\label{sec: top down mechanism}
In this section, we present \texttt{InfTDA}, but first let us explain why it is appropriate to use a TopDown approach to ensure hierarchical accuracy. 
The goal is to release a differentially private tabular data $\tilde{D}$, allowing the data analyzer to perform any marginal query. 
Under bounded DP, the total number of users, denoted as $n$, remains fixed within neighboring datasets; thus, it may be disclosed without compromising privacy.
However, when applying a differentially private mechanism directly to the histogram representation of $D$ (that is, at the higher level of the hierarchy), the perturbed aggregated $\tilde{n}$ tends to vary around $n$, with its variance increasing proportionally to the number of point queries. 
Consider the case  $X = X'$ , where the origin and destination geographic spaces are the same. When releasing an O/D dataset involving  $k$  areas using the Gaussian mechanism with a constant privacy budget, the variance of  $\tilde{n}$  is  $\text{Var}(\tilde{n}) = O(k^2)$, as the potential number of O/D pairs is  $k^2$. Due to the cancellation effect of the unbiased estimates produced by the Gaussian mechanism, the maximum error becomes $O(k)$. In contrast, the Stability Histogram does not produce unbiased estimates, resulting in a maximum error of $O(k^2)$. This reasoning applies at any level of the hierarchy, providing a theoretical upper bound on accuracy that does not satisfy the hierarchical accuracy effect we aim to achieve.

\begin{proposition}[Maximum Absolute Error per Level for Baselines]
\label{proposition: baselines}
Given a non-negative hierarchical tree $\mathcal{T}$ with branching factor $b$, depth $T$, and a parameter $\beta\in (0,1)$. The application of the $\rho$-zCDP Gaussian mechanism at level $T$ achieves for any $\ell \in [T]$ with probability at least $1-\beta$
\begin{equation}
\label{eq: VanillaGauss}
    \max_{u_{\ell}\in V_{\ell}}|\err (u_{\ell})|\leq O\bigg(b^{\frac{T-\ell}{2}}\sqrt{\frac{\ell}{\rho}\log\bigg(\frac{b}{\beta}\bigg)}\bigg).
\end{equation}
While, for $\beta = n\delta$, the application of the $(\varepsilon, \delta)$-DP Stability Histogram mechanism at level $T$ achieves with probability at least $1-\beta$
\begin{equation}
\label{eq: VanillaSH}
    \max_{u_{\ell}\in V_{\ell}}|\err (u_{\ell})|\leq O\bigg(\frac{\min(b^{T-\ell}, n)\log(1/\delta)}{\varepsilon}\bigg).
\end{equation}
\end{proposition}
The proof of the Proposition can be found in Appendix~\ref{app: additional proof}.
To solve this problem, we could compute the DP estimates of each range query by reallocating the privacy budget among the geographic levels. 
However, this approach returns inconsistent information about the dataset. For example, the computed flow between two regions might appear smaller than the aggregate flows between their constituent cities.
This issue can be addressed by reconciling the estimates with the closest possible values that satisfy certain consistency constraints, as done in the \texttt{CBDP} mechanism \cite{fioretto2018constrained} and the Hierarchical mechanism \cite{Boosting_accuracy}.
However, the first solves a unique optimization problem for the entire set of queries, yielding a solution that does not scale well, while the latter may return queries with negative values.

We propose a different approach, based on \texttt{TDA} developed by the US Census \cite{abowd20222020}. For a non-negative hierarchical tree, we iterate a differentially private algorithm starting from the root. At each iteration, an optimization problem is solved using information from the previous level. 
Unlike \texttt{TDA}, which employs an optimization with an $\ell_2$ objective function, our optimization algorithm minimizes an $\ell_\infty$ objective function, specifically the Chebyshev distance to the noisy vector. 
We demonstrate both theoretically and experimentally that this is a valid alternative to ensure hierarchical accuracy.

In contrast to \texttt{TDA}, where $\ell_2$ minimization leads to a unique solution, minimizing the Chebyshev distance yields multiple optimal solutions.
In Section \ref{section: IntOpt} we developed an algorithm for integer-constrained optimization that returns an optimal solution that effectively reduces false positives. 
Another advantage of this approach is that it operates entirely in the integer domain. By comparison, \texttt{TDA} constrained optimization introduced in \cite{abowd20222020} proceeds in two phases: first, it solves the constrained optimization in the real domain (a relaxation of the integer problem), and then performs a secondary optimization to determine the best rounding. We now present in detail our algorithm.

\subsection{\texttt{InfTDA}: TDA with Chebyshev Distance}
\begin{algorithm}[t]
\caption{\texttt{InfTDA}}\label{algo: InfTDA}
\begin{algorithmic}[1]
\Require Tree $\mathcal{T}=(V,E)$ of depth $T$, privacy budget $\rho>0$.
\State $\tilde{V}_0\gets \{(u_0, n)\}$
\For{$\ell \in (1,\dots, T)$} 
    \State $\tilde{V}_\ell \gets \{\}$ \Comment{DP nodes at level $\ell$}
    \For {$(u, c) \in \tilde{V}_{\ell-1}$} \Comment{Go through the constraints}
        \vspace{2pt}
        \State $\tilde{\bf q} \gets {\bf q}_{\mathcal{C}}(u)+\mathcal{N}_{\mathbb{Z}}\big(0,T/\rho\big)^{\dim({{\bf q}_{\mathcal{C}}(u)})}$\Comment{Apply noise}
        \vspace{2pt}
        \State $\bar{\bf q} \gets \texttt{IntOpt$_\infty$}\big(\tilde{\bf q}\,,\, c\big)$ \Comment{Solve optimization}
        \vspace{2pt}
        \State $C \gets \mathcal{C}(u)$ \Comment{Collect set of child nodes of $u$}
        \State $X \gets \{(C_j, \bar{q}_j)\,:\, \bar{q}_{j}>0\}$ \Comment{Drop zero attributes} 
        \State $\tilde{V}_{\ell} \gets \tilde{V}_{\ell} \cup X$ \Comment{Update level}
    \EndFor
\EndFor\\
\Return $\tilde{\mathcal{T}} \gets \big(\cup_{\ell=0}^{T}\tilde{V}_\ell, E\big)$\Comment{DP Tree}
\end{algorithmic}
\end{algorithm}
The TopDown Gaussian Optimized Mechanism with Chebyshev distance optimization \texttt{InfTDA}, operates on the non-negative hierarchical tree, such as the destination tree introduced in Section \ref{sec:tree}.
Similarly to \texttt{TDA}, this method applies discrete Gaussian noise to each level of the tree in a TopDown way, followed by a constrained optimization procedure before descending to the next level.

Since we use bounded privacy, the root attribute, representing the total number of users $n$ in the dataset, can be released without compromising privacy. However, if unbounded privacy were required, perturbing the root attribute would become necessary.
The algorithm then perturbs the attributes at the first level of the tree using a discrete Gaussian mechanism. 
The resulting vector is then post-processed to satisfy the hierarchical consistency and non-negativity constraints by solving an integer optimization problem.
For each optimized node, the algorithm selects its child nodes, applies a discrete Gaussian mechanism to their attributes, and optimizes them to ensure they are non-negative integers that sum to the attribute of the parent node. The procedure is executed iteratively until the final level $T$ is reached and optimized.

The detailed pseudocode of \texttt{InfTDA} is provided in Algorithm \ref{algo: InfTDA}. The process begins by constructing the root $\tilde{V}_0$ of the differentially private tree in line 1. Here, $u_0$ denotes the root node of the input tree, while $n$ represents its attribute. The algorithm then starts the TopDown loop in line 2. 
Each iteration aims to construct the set of nodes at level $\ell$ of the differentially private tree, instantiated in line 3.
Each node of the previous level $\ell-1$ is sampled in line 4 and used as a constraint. 
In line 5 the discrete Gaussian mechanism with $\rho/T$ privacy budget (for zCDP) is applied to the attributes ${\bf q}_{\mathcal{C}}(u)$ of the child nodes of the constraint. 
Then, in line 6 the private attributes $\tilde{\bf{q}}$ are post-processed to satisfy the constraints. 
The algorithm \texttt{IntOpt} solves the following integer optimization problem by minimizing the Chebyshev distance
\begin{align}
\label{eq: optimization}
	\mathcal{P}({\bf x}, c) &:= \arg \min_{\bf{y}}||{\bf x}-{\bf y}||_{\infty}\\
	&\quad \text{s.t.} \quad y_{i}\in \mathbb{N}_{0} \qquad \forall i\in[1,\dots,b]\notag\\
	&\quad \text{s.t.} \quad \sum_{i=1}^{b}y_i = c.\notag
\end{align}
The algorithm \texttt{IntOpt} is described in Section \ref{section: Integer Optimization}. 
In line 7, the set of child nodes of the constraint is constructed, and it is augmented with the corresponding post-processed DP attributes in line 8, dropping nodes with zero attributes. 
This final step effectively reduces the size of the DP tree and the running time of the algorithm, particularly for sparse datasets.
In fact, if there is a node $u_\ell$ with optimized attribute $\bar{q}(u_\ell)=0$, then, by consistency, the entire branch of the tree starting at $u_\ell$ will also have nodes with zero attributes.
Finally, in line 9 the set of DP nodes is updated, and this set will serve as constraints in the next iteration.
The algorithm outputs a differentially private tree, with optimized attributes. Note that the leaves of the DP tree constitute the histogram representation of differentially private tabular data.

\begin{theorem}[Privacy of \texttt{InfTDA}]
\texttt{InfTDA} satisfies $\textnormal{GS}_2(q)^2\tfrac{\rho}{2}$-zCDP under bounded privacy.
\end{theorem}
\begin{proof}
The attribute of each node at an even level of the tree represents a cross-range query of order one, while nodes at odd levels contain intra-level queries.  Hence, an entire level is a histogram of cross or intra-level queries, which have general $2$-global sensitivity $\textnormal{GS}_2(q)$.
Theorem \ref{th: discrete gaussian} implies that each iteration of the TopDown loop uses $\textnormal{GS}_2(q)^2\tfrac{\rho}{2T}$ privacy budget. As the loop goes through $T$ levels, by the composition and post-processing properties, the algorithm satisfies $\textnormal{GS}_2(q)^2\tfrac{\rho}{2}$-zCDP.
\end{proof}

\paragraph{Different Privacy Types and Sensitivities} If each user in the dataset contributes $m$ \emph{distinct} trips, the $2$-global sensitivity becomes $\sqrt{2m}$ for bounded privacy, and $\sqrt{m}$ for unbounded privacy. In cases where each user contributes $m$ trips without requiring them to be distinct, the $2$-global sensitivity becomes $\sqrt{2}m$ for bounded privacy and $m$ for unbounded privacy.  
An algorithm that takes into account unbounded privacy is provided in Appendix~\ref{app: inftda unbounded}. Essentially, it privatizes $n$ in $\tilde{V}_0$ and rescales the variance of the Gaussian mechanism.
Our experiments focus on bounded privacy with $m=1$.
\begin{corollary}
\label{corollary: m=1}
When each user in the O/D dataset used to construct the tree contributes a single trip, 
\texttt{InfTDA} satisfies $\rho$-zCDP under bounded privacy.
\end{corollary}

We now provide an upper bound for the maximum absolute error for each level of the tree.
\begin{theorem}[Utility of \texttt{InfTDA}] 
\label{th: InfTDA}
Given a non-negative hierarchical tree $\mathcal{T}$ with branching factor $b$, depth $T$, and a parameter $\beta \in (0,1)$. For each level $\ell \in [1,\dots, T]$, \texttt{InfTDA} with privacy budget $\rho>0$ achieves with probability at least $1-\beta$
\begin{equation*}
    \max_{u_\ell \in V_{\ell}}|\err(u_\ell)|\leq O\left(\sqrt{\frac{\ell^3 T}{\rho}\log\left(\frac{b\ell}{\beta}\right)}\right)
\end{equation*}
\end{theorem}
\begin{proof} The algorithm applies Gaussian noise in a TopDown way to all attributes of the nodes, except the root. 
Then, it performs an optimization procedure at each level. 
Consider a node $u_\ell$ at level $\ell\in [1, \dots, T]$, with attribute $q(u_\ell)$. Let $\tilde{q}(u_\ell)$ be the attribute returned by the Gaussian mechanism before the optimization is applied, and $\bar{q}(u_\ell)$ after the optimization. By triangle inequality we have 
\begin{equation}
\label{eq: proof InfTDA 1}
    |\err(u_\ell)|=|q(u_\ell)-\bar{q}(u_\ell)| \leq |q(u_\ell)-\tilde{q}(u_\ell)|+|\tilde{q}(u_\ell)-\bar{q}(u_\ell)|.
\end{equation}
The first term is just the absolute value of a Gaussian random variable, so we focus on the second term. Let $u_{\ell-1}$ be the father node of $u_\ell$, then $\bar{q}(u_\ell) $ is an element of the vector solution to the optimization problem ${\bf \bar{q}}_{\mathcal{C}}(u_{\ell-1})=\mathcal{P}({\bf \tilde{q}}_{\mathcal{C}}(u_{\ell-1}), c)$, where $c=\bar{q}(u_{\ell-1})$, thus
\begin{equation*}
    |\tilde{q}(u_\ell)-\bar{q}(u_\ell)|\leq ||{\bf \tilde{q}}_{\mathcal{C}}(u_{\ell-1})-{\bf \bar{q}}_{\mathcal{C}}(u_{\ell-1})||_\infty.
\end{equation*}
 Let us consider another vector $\boldsymbol{\xi}$, called the \emph{offset}, such that ${\bf q}_{\mathcal{C}}(u_{\ell-1}) + {\boldsymbol{\xi}}$ lies within the feasible region of the constrained optimization problem in Equation \ref{eq: optimization}, then
 \begin{align}
  \label{eq: proof constrain 1}
 q_{\mathcal{C}, j}(u_{\ell-1})+\xi_j &\geq 0\\
\label{eq: proof constrain 2}
 \sum_{j=1}^b q_{\mathcal{C}, j}(u_{\ell-1})+\xi_j &= \bar{q}(u_{\ell-1}) .
\end{align}
 As the vector ${\bf \bar{q}}_{\mathcal{C}}(u_{\ell-1})$ is a solution to the optimization problem, it minimizes the Chebyshev distance with ${\bf \tilde{q}}_{\mathcal{C}}(u_{\ell-1})$ under the non-negativity and summation constraints, then by triangle inequality
 \begin{align*}
     ||{\bf \tilde{q}}_{\mathcal{C}}(u_{\ell-1})-{\bf \bar{q}}_{\mathcal{C}}(u_{\ell-1})||_\infty&\leq ||{\bf \tilde{q}}_{\mathcal{C}}(u_{\ell-1})-{\bf q}_{\mathcal{C}}(u_{\ell-1})-\boldsymbol{\xi}||_\infty \\
     &\leq ||{\bf \tilde{q}}_{\mathcal{C}}(u_{\ell-1})-{\bf q}_{\mathcal{C}}(u_{\ell-1})||_\infty + ||\boldsymbol{\xi}||_\infty.
 \end{align*}
As $|q(u_\ell)-\tilde{q}(u_\ell)|\leq ||{\bf \tilde{q}}_{\mathcal{C}}(u_{\ell-1})-{\bf q}_{\mathcal{C}}(u_{\ell-1})||_\infty$, the upper bound in Equation~\ref{eq: proof InfTDA 1} becomes
\begin{equation}
\label{eq: proof continue}
    |\err(u_\ell)|\leq 2 ||{\bf \tilde{q}}_{\mathcal{C}}(u_{\ell-1})-{\bf q}_{\mathcal{C}}(u_{\ell-1})||_\infty + ||\boldsymbol{\xi}||_\infty.
\end{equation}
The problem is now to find an upper bound for $||\boldsymbol{\xi}||_{\infty}$.

To Upper bound $||\boldsymbol{\xi}||_{\infty}$, we now construct an example of $\boldsymbol{\xi}$ satisfying the constraints and having a bounded $\ell_\infty$ norm. By construction, from Equation \ref{eq: proof constrain 2} we have that 
\begin{align*}
    \sum_{j=1}^{b}\xi_j = \bar{q}(u_{\ell-1})-\sum_{j=1}^{b}q_{\mathcal{C}, j}(u_{\ell-1}) &= \bar{q}(u_{\ell-1})-q(u_{\ell-1}) \\
    &= \err(u_{\ell-1}).
\end{align*}
If $\err(u_{\ell-1})\geq 0$ we can take $\xi_j = \frac{\err (u_{\ell-1})}{b}$ for any $j\in[b]$ as a solution satisfying the summation constraint and the inequality constraint in Equation~\ref{eq: proof constrain 1}. However, this is not sufficient. If $\err(u_{\ell-1})< 0$ the inequality constraint might be not satisfied. In this scenario we might consider a solution where $\xi_i = 0$ for any $i \in [b]\setminus\{i^*\}$ where $\xi_{i^*} = -|\err(u_{\ell-1})|$. Any zero element satisfies the constraint in Equation~\ref{eq: proof constrain 1} as $q_{\mathcal{C}, j}(u_{\ell-1})\geq 0$. If $\xi_{i^*} \geq -q_{\mathcal{C}, i^*}(u_{\ell - 1})$ then we finish and obtain an upper bound $||\boldsymbol{\xi}||_{\infty}\leq |\err(u_{\ell-1})|$. In the other case where $\xi_{i^*}<-q_{\mathcal{C}, i^*}(u_{\ell-1})$ we need to augment $\xi_{i^*}$ up to meet $-q_{\mathcal{C}, i^*}(u_{\ell-1})$. By doing so we increase $\sum_i \xi_i$ making necessary to decrease some elements of the offset. As we are reducing elements that initially are zero, the new offset still contains only negative elements, and as  $\sum_i \xi_i = -|\err(u_{\ell-1})|$ any element cannot be less than $-|\err(u_{\ell-1})|$. Thus, we conclude that there always exists an offset such that $||\boldsymbol{\xi}||_{\infty}\leq |\err(u_{\ell-1})|$.

Now we continue from the upper bound in Equation \ref{eq: proof continue}
\begin{equation*}
    |\err(u_\ell)|\leq 2 ||{\bf \tilde{q}}_{\mathcal{C}}(u_{\ell-1})-{\bf q}_{\mathcal{C}}(u_{\ell-1})||_\infty + |\err(u_{\ell-1})|.
\end{equation*}
Completing the recurrence relation by ending at $\err(u_0) = 0$ we get
\begin{equation}
\label{eq: proof InfTDA 2}
     |\err(u_\ell)| \leq 2\sum_{\kappa = 1}^{\ell-1} ||{\bf \tilde{q}}_{\mathcal{C}}(u_{\ell-\kappa})-{\bf q}_{\mathcal{C}}(u_{\ell-\kappa})||_\infty.
\end{equation}
For $\ell=0$, we have $\err(u_0)=0$ under bounded privacy, while for $\ell=1$ the error is $|\err(u_{1})|\leq 2 ||{\bf \tilde{q}}_{\mathcal{C}}(u_0)- {\bf q}_{\mathcal{C}}(u_0)||_\infty$ which is twice the deviation caused by adding Gaussian noise alone. In Equation \ref{eq: proof InfTDA 2} we sum the $\ell_\infty$ norms of Gaussian random vectors with zero mean and variance $T/\rho$, hence by applying the tail bound in Corollary \ref{corollary: discrete Gaussian tail} and a union bound over the dimension $b$ of each vector and $\ell$ levels, we get for any $\beta \in (0,1)$
\begin{equation*}
    \text{Pr}\left[\max_{\kappa \in [\ell]}||{\bf \tilde{q}}_{\mathcal{C}}(u_{\kappa})-{\bf q}_{\mathcal{C}}(u_{\kappa})||_\infty \geq O\left(\sqrt{\frac{T}{\rho}\log\left(\frac{b\ell}{\beta}\right)}\right)\right]\leq \beta.
\end{equation*}
Thus, Equation \ref{eq: proof InfTDA 2} yields the following upper bound with probability at least $1-\beta$
\begin{equation*}
    |\err(u_\ell)| \leq O\left(\ell \sqrt{\frac{T}{\rho}\log\left(\frac{b\ell}{\beta}\right)}\right).
\end{equation*}
The claim follows by a union bound over $b^{\ell}=|V_{\ell}|$ nodes at level $\ell$, leading to an additional $\sqrt{\ell}$ factor.
\end{proof}
\subsection{\texttt{IntOpt}: Integer Optimization with Chebyshev Distance}
\label{section: IntOpt}
\label{section: Integer Optimization}
\begin{algorithm}[t]
	\caption{ $\ell_\infty$ Integer Optimization (\texttt{IntOpt$_\infty$})}\label{algo: IntOpt}
	\begin{algorithmic}[1]
		\Require ${\bf x} \in \mathbb{Z}^d, c\in \mathbb{N}$.
		\State ${\bf z} \gets \max\big(\big\lceil \frac{c-\sum_i x_i}{d}\big\rceil , -{\bf x}\big)$
        \State $t \gets ||\bf{z}||_{\infty}$
        \State $I \gets \text{sorted indices of ${\bf x}$ in ascending order}$
        \State $j\gets 0$
        \While{$\sum_i z_i > c-\sum_i x_i$}
            \State $\Delta \gets \sum_i z_i - c+\sum_i x_i$
            \State $z_{I[j]} \gets \max(z_{I[j]}-\Delta, -x_{I[j]}, -t)$
            \State $j \gets (j + 1) \mod |I|$
            \If {$j = 0$}
                \State $t \gets t + 1$
            \EndIf
        \EndWhile
        \State \bf{return } ${\bf x}+ {\bf z}$
	\end{algorithmic}
\end{algorithm}

In this section, we present an algorithm to solve the integer optimization problem aimed at minimizing the $\ell_\infty$ norm, with special attention to reducing false positives. Given a vector of integers ${\bf x} \in \mathbb{Z}^{b}$, representing the output of a differentially private mechanism, and a natural number $c \in \mathbb{N}$, the integer optimization problem $\mathcal{P}({\bf x}, c)$, as defined in Equation~\ref{eq: optimization}, can be reformulated by introducing ${\bf z} = {\bf y} - {\bf x}$. Minimizing $||{\bf z}||_{\infty}$ is equivalent to solving the following linear problem
 \begin{align}
\label{eq: linear optimization}
	\min \alpha &\quad \text{s.t.} \quad -\alpha \leq z_i \leq \alpha & \forall i\in[1,\dots, b]\notag\\
    &\quad \text{s.t.} \quad z_i \geq -x_i & \forall i\in[1,\dots, b]\\
	&\quad \text{s.t.} \quad \sum_{i=1}^{b}z_i = c-\sum_{i=1}^{b}x_i.&\notag
\end{align}
With this reformulation, we can compute a lower bound for the minimum $\alpha^*=||{\bf z}||_{\infty}$ satisfying the constraints of the problem in Equation \ref{eq: linear optimization}. 
 \begin{lemma} Let $\alpha^*$ be the solution of the linear program in \autoref{eq: linear optimization}, then
 \begin{equation}
 \label{eq: lower bound}
     \alpha^* \geq \max\left(\left\lceil \left|\frac{c-\sum_{i=1}^{b}x_i}{b}\right|\right\rceil\,;\, -\min_{i} x_i \right).
 \end{equation}
 \end{lemma}
 \begin{proof}
     From the constraint $-x_i \leq z_i \leq \alpha$, it follows that $\alpha \geq -x_i$ for all $i \in [b]$. Thus, $\alpha$ must satisfy $\alpha \geq \max_i(-x_i) = -\min_i x_i$. Additionally, the equality constraint, combined with $-\alpha \leq z_i \leq \alpha$, implies that $\alpha \geq \big|\frac{c - \sum_i x_i}{b}\big|$. Therefore, we deduce $$\alpha \geq \max\left(\left|\frac{c - \sum_i x_i}{b}\right|, -\min_i x_i\right),$$
     for the relaxed problem in the real domain. Since the feasible region of the relaxed real problem includes the feasible region of the integer problem, the final value of $\alpha$ is obtained by applying the ceiling function.
 \end{proof}
 Note also that the solution is not generally unique. For instance, consider ${\bf x} = (0,-1,1)$ and $c=2$. In this case, there are two possible solutions ${\bf y}_{1} = (1,0,1)$ or ${\bf y}_2 = (0,0,2)$ both of which have a Chebyshev distance of $1$ from ${\bf x}$.
In Algorithm~\ref{algo: IntOpt} we present a simplified version of our optimizer. A faster implementation, guaranteed to run in polynomial time with respect to $b$, can be found in Appendix~\ref{appendix: fast int opt}. 
The core idea is to initialize a solution that satisfies the inequality constraints, achieves a small $\ell_\infty$ norm, and has a summation exceeding the required value. The algorithm then iteratively reduces the entries of the solution to meet the summation constraint while minimizing any increase in the objective function.

\begin{lemma}[Optimality] Algorithm $\ref{algo: IntOpt}$ returns a solution that minimizes the Chebyshev distance.
\end{lemma}
\begin{proof}
In line 1, we propose our initial solution. First, we demonstrate that its summation exceeds the required value:
\begin{equation*}
    \sum_{i=1}^{b}\max\bigg(\bigg\lceil \frac{c-\sum_{i=1}^{b} x_i}{b}\bigg\rceil \,,\,-x_i\bigg)\geq  \sum_{i=1}^{b}\bigg\lceil \frac{c-\sum_{i=1}^{b} x_i}{b}\bigg\rceil\geq c-\sum_{i=1}^{b}x_i.
\end{equation*}
Next, we prove that $||{\bf z}||_{\infty}$ is at most equal to the lower bound in Equation~\ref{eq: lower bound}. When $c - \sum_i x_i < 0$, we obtain:
\begin{equation}
\label{eq: first proof}
    \max_i |z_i| \leq \max\bigg(\bigg\lceil \bigg|\frac{c-\sum_{i=1}^{d} x_i}{d}\bigg|\bigg\rceil\,,\, -\min_i x_i\bigg),
\end{equation}
which follows from $\big|\big\lceil\frac{c-\sum_i x_i}{d}\big\rceil\big| = \big\lceil\big|\frac{c-\sum_i x_i}{d}\big|\big\rceil - 1$ \footnote{For $a < 0$, we have $|\lceil a\rceil| = -\lceil a\rceil = \lfloor -a \rfloor = \lfloor |a|\rfloor = \lceil |a|\rceil - 1$.}. Conversely, when $c - \sum_i x_i \geq 0$, $||{\bf z}||_{\infty}$ matches the lower bound in Equation~\ref{eq: lower bound}.
Thus, the algorithm starts with a vector that satisfies the inequality constraints, has a small $\ell_\infty$ norm, and its entries sum to a value larger than what is required. 
In the loop from lines 5 to 10, each entry in the vector ${\bf z}$ is reduced iteratively until its total summation satisfies the constraint.
To guarantee optimality, it is crucial to ensure that no entry is excessively reduced, thus avoiding unnecessary increases in $||{\bf z}||_\infty$.
This consideration is addressed in line 2, where the algorithm identifies the smallest possible entry of $\mathbf{z}$ such that $||{\bf z}||_{\infty}$ remains unchanged. 
In line 7, the algorithm updates $z_i$. When set to $z_i - \Delta$, where $\Delta$ is defined in line 6 as the positive remainder, the process terminates, and the result ${\bf y} = {\bf x} + {\bf z}$ is returned. The update respects the inequality constraint $z_i \geq -x_i$ and the optimality condition $z_i \geq -t$. If a solution is achieved in the first round of updates, it is guaranteed to be optimal since $t$ corresponds to the lower bound given in \autoref{eq: lower bound}. Otherwise, $t$ increases by one, allowing smaller entries and thereby increasing $||{\bf z}||_\infty$ by 1, which is the minimum increase.
A solution and therefore the end of the cycle is always guaranteed as if no updates are possible then ${\bf z} = -{\bf x}$ and so $\sum_{i}z_i \leq c - \sum_{i}x_i$ for $c \geq 0$.
\end{proof}

\paragraph{Reducing False Positives} The updates in line 7 can be performed iteratively using any permutation $I$ of the indices of $\mathbf{z}$. In line 3, we propose a specific permutation. The idea behind this choice is that when $\mathbf{x}$ results from a differentially private mechanism, false positives are likely to be associated to attributes with small values. Consequently, the permutation in line 3 prioritizes reducing the smaller elements first, potentially setting them to zero (i.e., $y_i = 0$). Alternatively, an inverse approach can be taken in line 3 by sorting the indices of $\mathbf{x}$ in descending order, which would focus on reducing false negatives instead.
\section{Experiments}
\label{section: experiments}
\begin{figure*}[!t]
    \centering
    \includegraphics[width=0.8\linewidth]{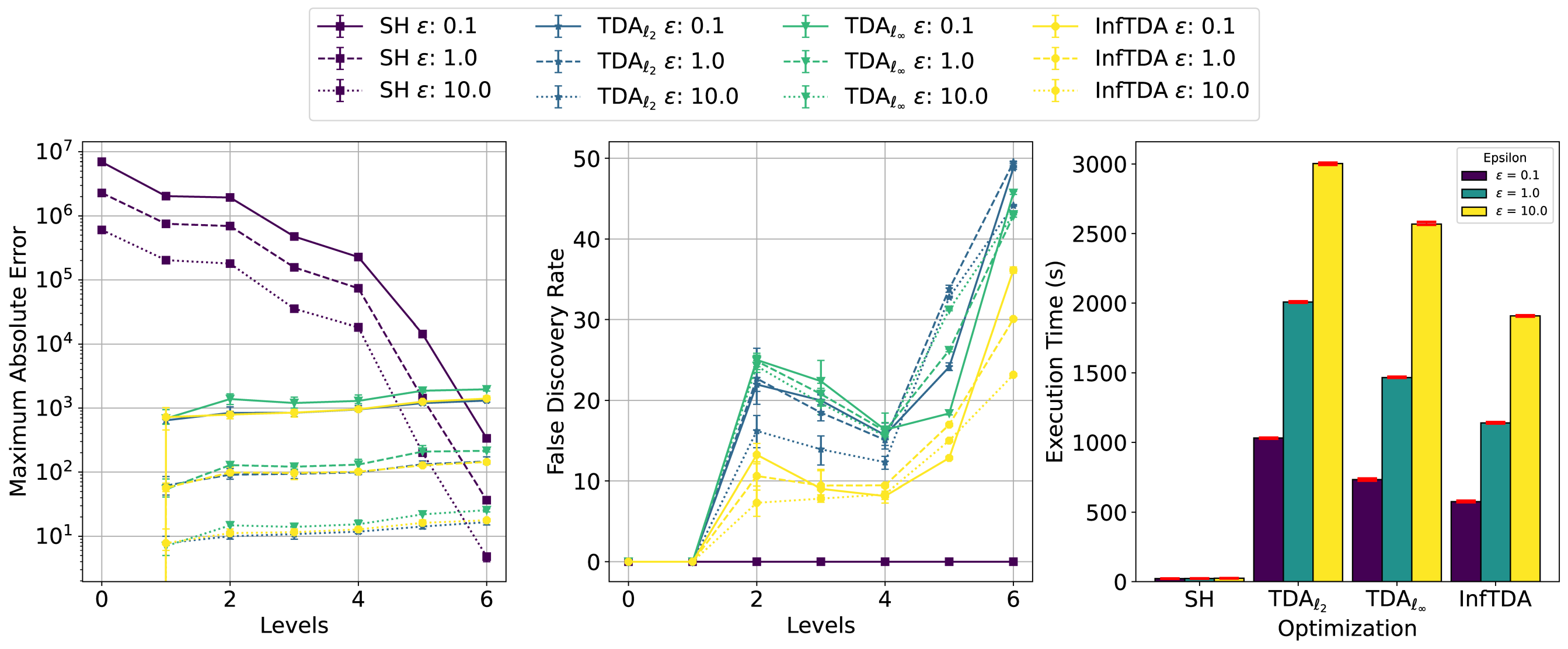}
    \caption{\small Experiments run for the Italian dataset (from ISTAT). From left to right: maximum absolute error, false discovery rate, and execution time. The error bars indicate maximum and minimum values over 10 experiments. 
    }
    \label{fig: italy experiments}
\end{figure*}

This section provides an experimental evaluation of \texttt{InfTDA} against various baselines in real-world and synthetic datasets. 
We start by presenting the baselines used, then we introduce the datasets and the experimental setup.

\paragraph{Baselines}
As simple baselines, we use the injection of discrete Gaussian noise, and the application of the Stability-Based histogram, on the leaves' attributes of the tree. They will be labeled respectively as \texttt{VanillaGauss} and \texttt{SH}. We are aware of more accurate baselines for sparse histograms \cite{desfontaines2022differentially, aumuller2021differentially}. However, in terms of maximum absolute error utility, they offer the same asymptotic performance as \texttt{SH}.
\texttt{VanillaGauss} offers stronger guarantees for range queries compared to \texttt{SH}, but it is space-inefficient and may produce negative counts. Conversely, \texttt{SH} generates datasets with non-negative counts and no false positives; however, it performs poorly for range queries, especially when the dataset consists predominantly of rare items. These baselines provide experimental evidence underscoring the necessity of using TopDown constraint optimization algorithms like \texttt{TDA} or our proposed \texttt{InfTDA}.
We evaluate \texttt{InfTDA} against two variations of TopDown algorithms, each employing a different optimization strategy. The first is a simple implementation of \texttt{TDA}, referred to as $\texttt{TDA}_{\ell_2}$, which uses the Euclidean distance as the objective function \footnote{We remark that the real implementation of \texttt{TDA} would require weighted non-negative least squares optimization, and a complex two step optimization to obtain integer values.}. The second, denoted as \texttt{TDA}$_{\ell_\infty}$, incorporates Chebyshev optimization but relies on a black-box solver for the optimization process. This latter baseline allows us to assess the effectiveness of our optimizer, \texttt{IntOpt}, particularly in reducing false positives.  
We recognize the possibility of employing \texttt{CBDP} and the Hierarchical mechanism as a baseline. However, the first faces significant implementation obstacles due to the vast quantity of O/D pairs in the analyzed dataset, while the second does not work in the integer domain.

\paragraph{Real Dataset}
The real-world dataset under examination originates from the Italian National Institute of Statistics (ISTAT) \cite{istat}.The dataset contains 28,805,440 commutes for 2011, which are trips departing from the habitual residence and returning to the same place daily. These trips were obtained from survey questionnaires in which each individual was asked to provide the address of their usual place of study or work. As each individual contributes one trip to the dataset, we apply the privacy analysis outlined in Corollary \ref{corollary: discrete Gaussian tail}.
This dataset is structured with a geographic hierarchy consisting of five partitions: regions, provinces, municipalities, section areas, and census sections, with regions representing the highest level (i.e., the largest areas). The first three geographic partitions are shown in Figure \ref{fig:italy}.
The dataset exhibits significant sparsity, containing 362,292 census sections, which theoretically could result in over 100 billion possible O/D pairs. However, it only records 14,287,549 actual flows. Due to computational limitations, we considered only O/D pairs up to the municipality level for our experiments. This results in approximately 500,000 O/D pairs out of more than 65 million possible pairs. We generate the destination tree from this dataset, obtaining a tree of depth  $T = 6$.

\paragraph{Synthetic Datasets} We generate two types of synthetic partitions, a \emph{binary} partition and a \emph{random} partition, to create six distinct O/D datasets.
The binary partition consists of 8 hierarchical levels, where each area is iteratively divided into two smaller areas, resulting in a binary destination tree with a depth of $T=16$. In contrast, the random partition has 4 levels, where each area is randomly divided into $k$ smaller areas, with $k$ sampled uniformly from 2 to 10.
This latter approach simulates real-world scenarios where areas are partitioned unevenly, creating a destination tree with a variable structure.
The O/D flows are sampled from a continuous Pareto distribution ($\text{Pr}(x) \sim x^{-\beta}$, where $x$ represents the flow and $\beta > 0$) and rounded, a common pattern observed in mobility and social data \cite{han2011origin, alessandretti2020scales}. These flows are then assigned as attributes to the leaves of the generated trees.
To evaluate the mechanisms under varying levels of sparsity, we simulate three scenarios: \emph{complete}, where all leaves have positive attributes; \emph{dense}, where $50\%$ of the leaves are assigned positive attributes; and \emph{sparse}, where only $1\%$ of the leaves are assigned positives attributes. This allows for a comprehensive testing of the performance of the mechanisms under different sparsity conditions. The number of users generated and the number of O/D in the synthetic dataset can be found in Table \ref{tab: synthetic info}.
\begin{table}
    \centering
    \begin{tabular}{|c|c|c|}
    \hline
         \bf Dataset & \bf Number of users & \bf Number of O/D\\
         \hline
         Binary Complete & 1051271 & 65536 \\
         Binary Dense & 734688 & 32768 \\
         Binary Sparse & 23302 & 655 \\
         \hline
         Random Complete & 2019580 & 189225 \\
         Random Dense & 1003943 & 95612 \\
         Random Sparse & 67840 & 1892\\
         \hline
    \end{tabular}
    \caption{\small Characteristics of synthetic datasets.}
    \label{tab: synthetic info}
\end{table}

\paragraph{Experimental Setup}
The system was developed using open source libraries and Python 3.11. Our approach to differential privacy leverages the OpenDP library\footnote{https://github.com/opendp/opendp} \cite{gaboardi2020programming}, which includes implementations of the discrete Gaussian mechanism and the Stability Histogram. For black-box optimization, we use cvxpy \cite{agrawal2018rewriting, diamond2016cvxpy}. \texttt{TDA}$_{\ell_2}$ first minimizes the Euclidean distance for the relaxed program in the real domain using CLARABEL \cite{chen2023efficient}, then it rounds and redistributes the exceeding in a similar fashion of \texttt{InfTDA}, hence by prioritizing the elimination of small values. \texttt{TDA}$_{\ell_\infty}$ optimization works completely in the integer domain and uses the GLPK mixed integer optimizer.
The tests were conducted using an Intel Xeon Processor W-2245 (8 cores, 3.9GHz), 128GB RAM, and Ubuntu 20.04.3. 
The experiments and the code are publicly available\footnote{\href{https://github.com/aidaLabDEI/TDA_hierarchical}{https://github.com/aidaLabDEI/TDA\_hierarchical}}. Each mechanism is run $10$ times, and the error bars in the graphs indicate the maximum and minimum value of the metric considered.

\paragraph{Privacy budget} The experiments were carried out with $\varepsilon\in[0.1, 1, 10]$ and $\delta = 10^{-8}$ (sufficient to ensure that $\delta \ll 1/n$). Privacy budget $\varepsilon\leq 1$ provides strong real-world privacy in most
cases \cite{wood2018differential}, however, there are examples of deployments with larger values. The US Census of 2020 distributed the population data under $(17.14, 10^{-10})$-DP \cite{bureau2020disclosure}. Unfortunately, such a high privacy budget may give meaningless theoretical guarantees in the worst-case scenario\footnote{The suspicion of a user’s presence in the dataset for a \emph{strong} attacker—who knows everything about the dataset except their target—grows significantly for large values of $\varepsilon$. For instance, with $ \varepsilon = 10$  and an initial suspicion of $10\%$, the attacker’s suspicion increases to $99.9\%$ after observing the differentially private dataset. In contrast, for $\varepsilon = 1$, the suspicion would increase only to $23\%$. For further details, refer to \cite{desfontaines2019}.}, but still may offer good privacy under some specific and more realistic attacks.
Although we recognize that intelligent allocation of the privacy budget across levels can improve the utility of a differentially private dataset (as done by the US Census \cite{bureau2020disclosure}), we opted for a uniform distribution of the budget across all levels. We used Lemma \ref{lemma: from DP to zCDP} to compute the privacy budget for zCDP.

\paragraph{Quality of the DP Dataset}
To assess the quality of the DP dataset $\tilde{D}$ we measured two key indicators that can be computed at any level of the tree: the \emph{maximum absolute error} defined in Equation \ref{eq: max error for non-negative tree} and the \emph{false discovery rate}, which is the percentage of O/D pairs present in $\tilde{D}$ but absent in the original data. In tree notation, for any level $\ell \in [T]$ the false discovery rate is
\begin{equation*}
    f_d(\ell; \tilde{D}, D) = \frac{|\{u_\ell: \tilde{q}(u_\ell)>0 \wedge  q(u_\ell)=0\}|}{|\{u_\ell: \tilde{q}(u_\ell)>0\}|}\cdot 100\%,
\end{equation*} 
where $D$ is the original data.

\begin{figure*}[t]
    \centering
    \includegraphics[width=0.8\linewidth]{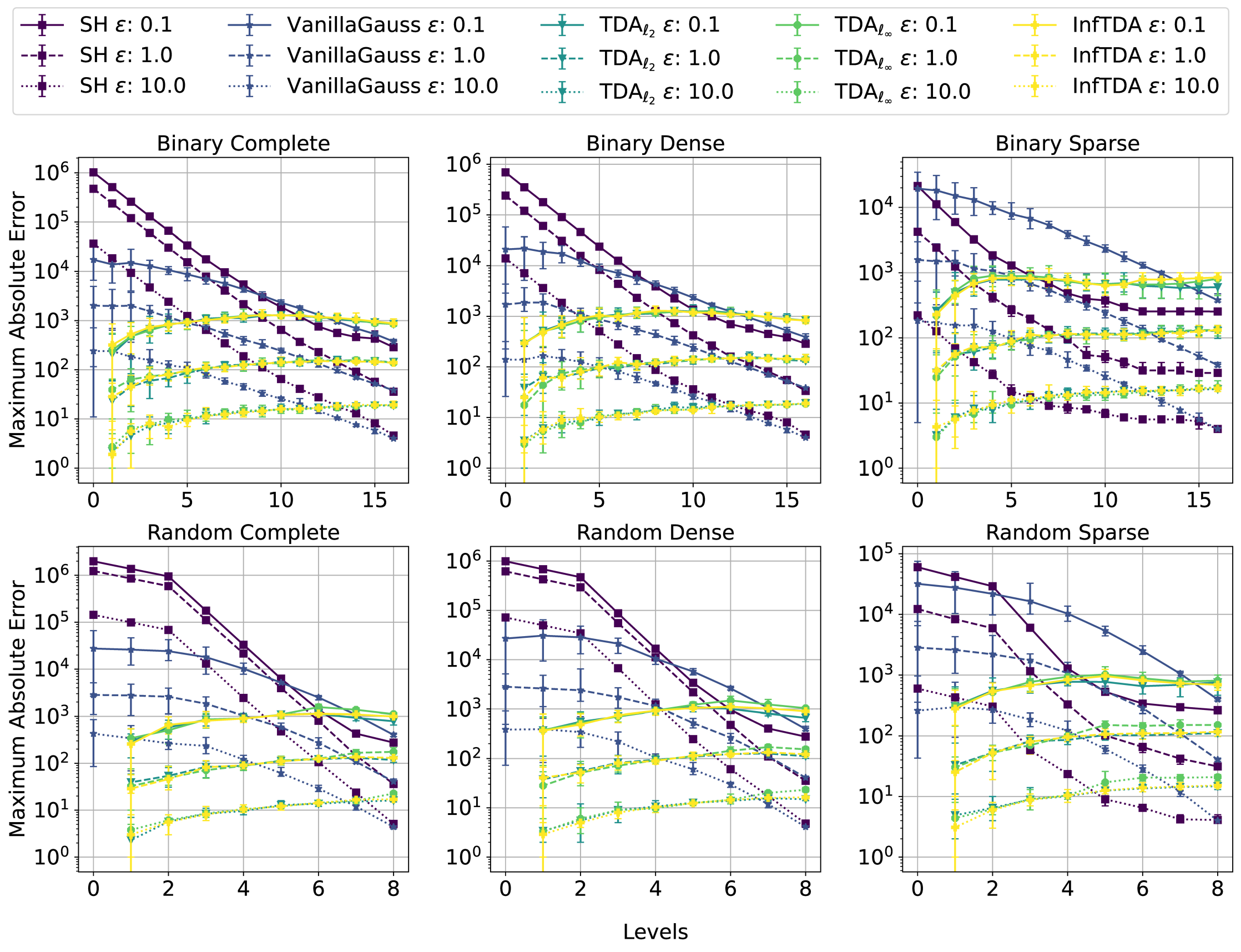}
    \caption{\small Experiments run for the synthetic datasets with a focus on maximum absolute error. The error bars indicate maximum and minimum values over 10 experiments}
    \label{fig: synthetic exp}
\end{figure*}

\subsection{Discussion}
\paragraph{Italian dataset}
The analysis of the O/D commuting dataset for Italy is presented in Figure \ref{fig: italy experiments}, which illustrates three metrics from left to right: the maximum absolute error per level, the false discovery rate per level and the running time. We emphasize that the levels are defined as follows: the zero level represents the total number of users, the second level corresponds to range queries for regions, the fourth level pertains to provinces, and the final level corresponds to municipalities' O/D commutes.
The \texttt{VanillaGauss} mechanism was not applied to this dataset due to the computational infeasibility of sampling more than 65 million Gaussian noises, one for each potential O/D commute. Instead, this baseline will be evaluated using synthetic data.

We immediately observe that \texttt{SH} performs poorly on range queries, although it provides the most accurate estimates for O/D flows at finer geographic levels. This is not surprising since the mechanism achieves on the leaves a maximum error close to a lower bound \cite{vadhan2017complexity} (Theorem 5.13).
Consistent with our theoretical findings in Theorem \ref{th: InfTDA}, \texttt{InfTDA} produces datasets with diminishing accuracy as one moves down the levels of the tree, reflecting improved precision for O/D flows at larger geographic scales.
This trend is similarly observed with \texttt{TDA}$_{\ell_2}$ and \texttt{TDA}$_{\ell_\infty}$. While \texttt{TDA}$_{\ell_2}$ shows comparable utility to \texttt{InfTDA}, \texttt{TDA}$_{\ell_\infty}$ underperforms slightly. This is due to the fact that Chebyshev optimization produces multiple optimal solutions, not all of which perform well in practice in terms of utility. This underscores the importance of our optimizer \texttt{IntOpt}, in selecting the most practical solution.
This is clearly illustrated in the middle plot of Figure \ref{fig: italy experiments}, where \texttt{InfTDA} consistently reduces false positive detections compared to \texttt{TDA}$_{\ell_\infty}$ and \texttt{TDA}$_{\ell_2}$, regardless of the privacy budget. The reduction in the false discovery rate depends on factors like hierarchical structure and attribute distribution, such as the presence of sparse levels followed by dense ones. Nonetheless, \texttt{InfTDA} consistently outperforms the other mechanisms in reducing false positives.

In terms of running time, \texttt{SH} is clearly the fastest, since it does not need to perform any kind of optimization and runs linearly with the number of users $n$. \texttt{InfTDA} is the fastest TopDown algorithm. Running time increases with the privacy budget because, in the low-privacy regime, false negatives are more frequent. These occur when attributes that were positive in the sensitive dataset appear as zero in the DP dataset. This reduction in positive attributes typically decreases the size of the returned dataset, leading to fewer optimizations. In conclusion, \texttt{SH} achieves significantly better accuracy, approximately an order of magnitude higher, for flows between municipalities, but performs poorly for range queries. Mechanisms based on \texttt{TDA} provide hierarchical accuracy and produce meaningful results even in high privacy regimes. For instance, for  $\varepsilon = 1$  (resp., $\varepsilon = 0.1$), they achieve a maximum absolute error of about 100 (resp., 1000). Additionally, \texttt{InfTDA} generates datasets with fewer false positives and operates more efficiently.

\paragraph{Synthetic Dataset} 
Figure \ref{fig: synthetic exp} illustrates the maximum absolute error for the six synthetic datasets we generated. Unlike the previous analysis, this evaluation includes \texttt{VanillaGauss}. As expected, \texttt{VanillaGauss} demonstrates better accuracy than \texttt{SH} for range queries, particularly when the dataset is not highly sparse. In contrast, \texttt{SH} performs better on sparse datasets; however, even in these cases, the error increases significantly for large range queries at the higher levels of the tree.
The TopDown algorithms exhibit similar behavior across all scenarios, with one exception: \texttt{TDA}$_{\ell_\infty}$ on the random sparse dataset for $\varepsilon\leq 1$. Although this deviation is not substantial compared to the other \texttt{TDA} algorithms, it provides valuable insights consistent with the findings from the Italian dataset. In trees with random branching factors, similar to the structure of the Italy tree, the number of variables in the optimization problem increases. 
This expansion of the solution space for the Chebyshev distance optimization increases the likelihood of sampling an optimal solution that performs poorly in practice. 
Conversely, this issue does not arise in binary trees, where each optimization involves only two variables, making it more likely to sample the best optimal solution for the Chebyshev minimization.

Figure \ref{fig: synthetic false discovery rate} shows the false discovery rate for the sparse datasets. For the binary tree, no significant improvements are observed, indicating that when optimization involves a small number of variables, the choice of the optimization function has minimal impact on the utility of the released datasets. In contrast, for the random tree, \texttt{InfTDA} generates a dataset with fewer false positives compared to the other mechanisms (with the obvious exception of \texttt{SH}).
Regarding the execution times of the algorithms, Figure \ref{fig: synthetic time} clearly shows that our optimizer is faster than a black-box optimizer. This is not only because \texttt{InfTDA} generates a dataset with fewer false positives but also because it runs faster even for the complete synthetic dataset.

\begin{figure}[t]
    \centering
    \includegraphics[width=0.9\linewidth]{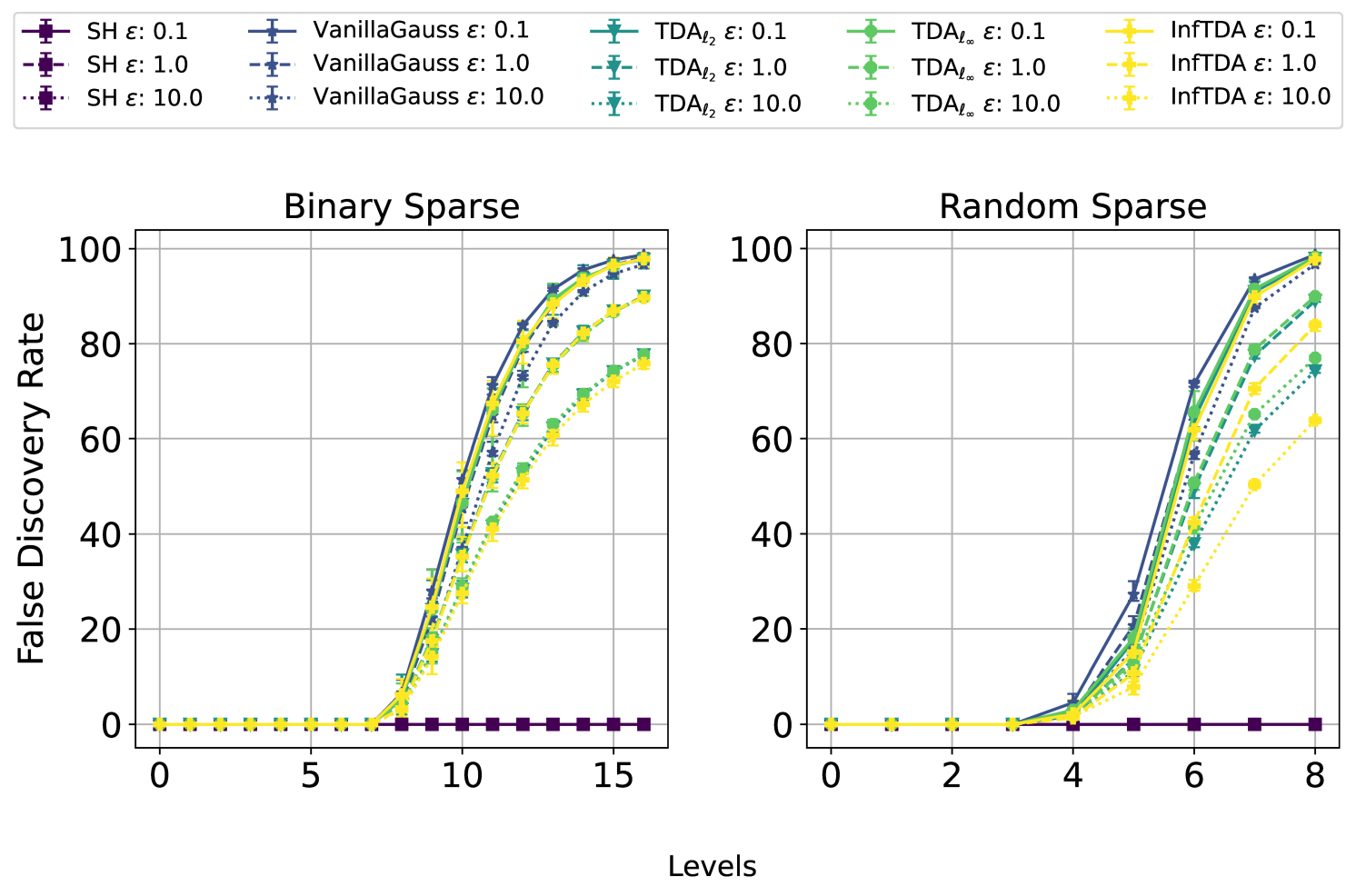}
    \caption{\small False discovery rate for synthetic datasets}
    \label{fig: synthetic false discovery rate}
\end{figure}
\begin{figure}[t]
    \centering
    \includegraphics[width=0.9\linewidth]{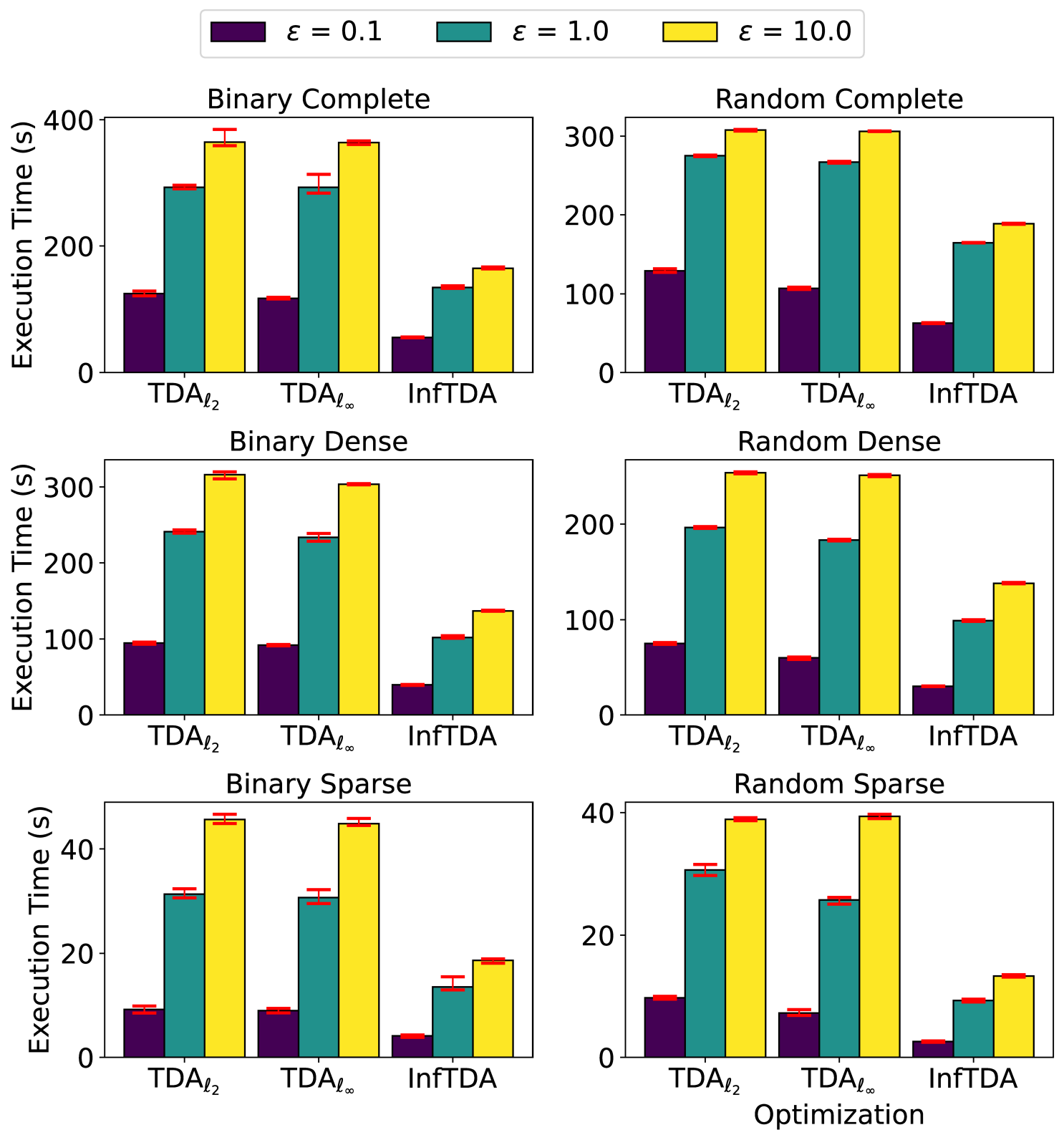}
    \caption{\small Running time for synthetic datasets and \texttt{TDA}}
    \label{fig: synthetic time}
\end{figure}
\section{Conclusion and Future Directions}
We found that a well-designed implementation of the TopDown algorithm can significantly improve the accuracy in the differentially private release of O/D data with a hierarchical structure. This improvement is particularly beneficial when producing tabular datasets where range queries must be more accurate for larger geographical areas. Specifically, this paper explores TopDown algorithms for general tree data structures with non-negative attributes and hierarchical consistency. We illustrate how any O/D dataset can be structured in the required tree format.
Additionally, we propose a Chebyshev distance constraint optimization problem as an alternative to the commonly used Euclidean distance minimization. This approach yields two key outcomes: a theoretical analysis of the maximum absolute error for a TopDown algorithm and a practical, efficient integer optimization algorithm that has been proven to reduce false positives on both real and synthetic datasets.
Our proposed TopDown algorithm, \texttt{InfTDA}, combined with our optimizer, \texttt{IntOpt}, outperforms naive baselines and is not worse than black-box implementations of $\texttt{TDA}$ with different objective functions on both real and synthetic datasets, while being simpler to implement, running faster and effectively reducing the false discovery rate.
Given the versatility of our approach, it would be valuable to test it on various real-world datasets beyond O/D data. Any tabular data, such as healthcare data, can be used to construct a non-negative hierarchical tree, provided that a query hierarchy is defined. We leave this exploration for future research. Furthermore, we believe that studying stability-based algorithms as a source of noise in a \texttt{TDA} would be interesting, particularly in scenarios where the hierarchy consists of unknown partitions. Finally, based on the results of our experiments, we believe that similar theoretical results hold for \texttt{TDA} with Euclidean distance minimization.

\begin{acks}
    This work was supported in part by the Big-Mobility project by the University of Padova under the Uni-Impresa call, by
the MUR PRIN 20174LF3T8 AHeAD project, and by MUR PNRR CN00000013 National Center for HPC, Big Data and
Quantum Computing.
\end{acks}

\bibliographystyle{ACM-Reference-Format}
\bibliography{main}

\appendix

\section{Additional Proof}
\begin{proof}[Proof of Proposition \ref{proposition: baselines}] 
\label{app: additional proof}
We start by proving the Equation \ref{eq: VanillaGauss}. By applying the discrete Gaussian mechanism on each attribute of the final level nodes we can reconstruct the attributes at any level using the hierarchical relation in equation \ref{eq: hierarchical relation}. Consider a node $u_\ell$ at level $\ell \in \{0,\dots, T-1\}$, its private attribute is
\begin{align*}
    \tilde{q}(u_\ell) = \sum_{i_{\ell+1}\in \mathcal{C}(u_\ell)} \dots \sum_{i_{T}\in \mathcal{C}(u_{T-1})}\tilde{q}(u_{T}),
\end{align*}
so the error $\err(u_\ell) = q(u_\ell)-\tilde{q}(u_\ell)$ is
\begin{equation}
\label{eq: appendix proof 1}
   \err(u_\ell) = \sum_{u_{\ell+1}\in \mathcal{C}(u_\ell)} \dots \sum_{u_{T}\in \mathcal{C}(u_{T-1})}\big[q(u_T)-\tilde{q}(u_{T})\big].
\end{equation}
Each error $q(u_T)-\tilde{q}(u_{T})$ at the $T$ level is a Gaussian random variable $\mathcal{N}_{\mathbb{Z}}(0, 1/\rho)$, then the right hand side of equation \ref{eq: appendix proof 1} is a Gaussian random variable $\mathcal{N}_{\mathbb{Z}}(0, b^{T-\ell}/\rho)$, as it is a summation of $b^{T-\ell}$ Gaussian random variables. Then, by applying Corollary \ref{corollary: discrete Gaussian tail} and a union bound over $b^{\ell}$ nodes, we obtain for $\beta \in (0,1)$
\begin{equation*}
    \text{Pr}\bigg[\max_{u_\ell \in V_\ell} |\err(u_\ell)|\leq O\bigg(b^{\frac{T-\ell}{2}}\sqrt{\frac{\ell}{\rho}\log\bigg(\frac{b}{\beta}\bigg)}\bigg)\bigg]\geq 1-\beta.
\end{equation*}

We now prove equation \ref{eq: VanillaSH}. Equation \ref{eq: appendix proof 1} is still applicable; however, the noise inserted at level $T$ follows a complicated distribution, which is the result of a symmetric Laplace noise distribution followed by a truncation. So it is not straightforward to compute the distribution of composition of such random variables. However, we have $|\err(u_T)|\leq O\big(\frac{\log(1/\delta)}{\varepsilon}\big)$ with probability at least $1-\delta$. Therefore, by summing the error and applying a union bound on $n$, since we have at most $n$ random variables to sum, equation \ref{eq: appendix proof 1} leads to $|\err(u_\ell)|\leq O\big(\min(b^{T-\ell},n)\frac{\log(1/\delta)}{\varepsilon}\big)$ with probability at least $1-n\delta$. As $\delta \ll 1/n$ follows a constant probability upper bound.
\end{proof}

\section{\texttt{InfTDA} for unbounded privacy}
\label{app: inftda unbounded}
\begin{algorithm}[t]
\caption{\texttt{InfTDA} for unbounded privacy}\label{algo: InfTDA unbounded}
\begin{algorithmic}[1]
\Require Tree $\mathcal{T}=(V,E)$ of depth $T$, privacy budget $\rho>0$.
\State $\tilde{n}\gets n + \mathcal{N}_{\mathbb{Z}}\left(0, \tfrac{T}{\rho}\right)$
\State $\tilde{V}_0\gets \{(u_0, \tilde{n})\}$
\For{$\ell \in (1,\dots, T)$} 
    \State $\tilde{V}_\ell \gets \{\}$ \Comment{DP nodes at level $\ell$}
    \For {$(u, c) \in \tilde{V}_{\ell-1}$} \Comment{Go through the constraints}
        \vspace{2pt}
        \State $\tilde{\bf q} \gets {\bf q}_{\mathcal{C}}(u)+\mathcal{N}_{\mathbb{Z}}\left(0, \tfrac{T}{\rho}\right)^{\dim({{\bf q}_{\mathcal{C}}(u)})}$\Comment{Apply noise}
        \vspace{2pt}
        \State $\bar{\bf q} \gets \texttt{IntOpt$_\infty$}\big(\tilde{\bf q}\,,\, c\big)$ \Comment{Solve optimization}
        \vspace{2pt}
        \State $C \gets \mathcal{C}(u)$ \Comment{Collect set of child nodes of $u$}
        \State $X \gets \{(C_j, \bar{q}_j)\,:\, \bar{q}_{j}>0\}$ \Comment{Drop zero attributes} 
        \State $\tilde{V}_{\ell} \gets \tilde{V}_{\ell} \cup X$ \Comment{Update level}
    \EndFor
\EndFor\\
\Return $\tilde{\mathcal{T}} \gets \big(\cup_{\ell=0}^{T}\tilde{V}_\ell, E\big)$\Comment{DP Tree}
\end{algorithmic}
\end{algorithm}

\begin{theorem}
\texttt{InfTDA} in Algorithm~\ref{algo: InfTDA unbounded} satisfies $\left(\frac{m^2}{T}+\textnormal{GS}_2(q)^2\right)\frac{\rho}{2}$-zCDP under unbounded privacy.
\end{theorem}
\begin{proof}
Each iteration of the TopDown loop consumes $\textnormal{GS}_2(q)^2\tfrac{\rho}{2T}$ privacy budget. Since the loop goes through $T$ levels, by composition and post-processing properties, the total privacy budget consumed is $\frac{\textnormal{GS}_2(q)^2 }{2}\rho$. Additionally, the privatization of the root attribute in line 1 increases the overall privacy budget. Given that $n$ represents a point query and the addition or removal of a single user alters $n$ by $m$, the $\ell_2$ sensitivity is $m$. Consequently, line 1 contributes an additional privacy budget of $\frac{m^2}{T}\rho$.
\end{proof}

\section{Fast $\texttt{IntOpt}_{\infty}$}
\label{appendix: fast int opt}
\begin{algorithm}[t]
	\caption{ Polynomial time $\ell_\infty$ Integer Optimization (\texttt{IntOpt$_\infty$})}\label{algo: Fast IntOpt}
	\begin{algorithmic}[1]
		\Require ${\bf x} \in \mathbb{Z}^b, c\in \mathbb{N}$.
		\State ${\bf z} \gets \max\big(\big\lceil \frac{c-\sum_{i}x_i}{b}\big\rceil , -{\bf x}\big)$
        \State $t \gets ||\bf{z}||_{\infty}$
        \State $I \gets \text{sorted indices of ${\bf x}$}$
        \State $I \gets (i\in I : z_i > -x_i)$
        \State $j\gets 0$
        \While{$\sum_i z_i > c-\sum_i{x_i}$}
            \State $\Delta \gets \sum_i z_i - c+\sum_i x_i$
            \State $z_{I[j]} \gets \max(z_{I[j]}-\Delta, -x_{I[j]}, -t)$
            \State $j \gets (j + 1) \mod |I|$
            \If {$j = 0$}
                \State $I \gets (i\in I : z_i > -x_i)$
                \State $r \gets \big \lfloor \frac{1}{|I|}\big(\sum_{i}z_i -c + \sum_{i}x_i\big) \big\rfloor$
                \State $t \gets t + \max(1, r)$
            \EndIf
        \EndWhile
        \State \bf{return } ${\bf x}+ {\bf z}$
	\end{algorithmic}
\end{algorithm}
This is the algorithm that we implemented in our experiments. One key modification is that the while loop operates only on the indices of elements in $\mathbf{z}$ that can still be reduced, as specified in lines 4 and 11. The main modification is in line 12, where it is computed the smallest entry of $\bf z$ that, if $z_i = -t-r$ for any $i \in I$, then we would have
\begin{equation*}
    c-\sum_{i=1}^{d}x_i \leq \sum_{i=1}^{b}z_i \leq  c-\sum_{i=1}^{b}x_i + |I|,
\end{equation*}
ensuring that only one additional reduction round is required to satisfy the equality constraint. If all elements are clipped to $-t - r$, the total summation decreases by $\sum_i z_i - r \cdot |I|$. The first inequality is derived as follows
\begin{align*}
    \sum_{i=1}^{b}z_i - r\cdot |I| &= \sum_{i=1}^{b}z_i - \bigg \lfloor \frac{1}{|I|}\big(\sum_{i}z_i -c + \sum_{i}x_i\big) \bigg\rfloor\cdot |I|\\
    &\geq \sum_{i=1}^{b}z_i - \frac{1}{|I|}\bigg(\sum_{i=1}^{b}z_i-c+\sum_{i=1}^{b}x_i\bigg)\cdot |I| = c-\sum_{i=1}^{b}x_i,
\end{align*}
while the second comes from
\begin{align*}
    \sum_{i=1}^{b}z_i - r\cdot |I| &\leq \sum_{i=1}^{b}z_i - \bigg[\frac{1}{|I|}\bigg(\sum_{i=1}^{b}z_i-c+\sum_{i=1}^{b}x_i\bigg)-1\bigg]\cdot |I|\\
    &= c-\sum_{i=1}^{b}x_i + |I|,
\end{align*}
If not all elements are clipped to $-t - r$, at least one element is clipped to $-x_i$, reducing the cardinality of $I$ by at least 1. Consequently, after the first reduction loop (line 10), the algorithm either transitions to the second-to-last loop (where all elements are clipped to $-t - r$) or reduces $|I|$ to $|I| - 1$. In the worst-case scenario, where $|I| = b$, it takes $O(b)$ iterations for $|I|$ to decrease by 1. Thus, the overall runtime of the algorithm is bounded by $O(b^2)$.

\end{document}